\documentclass[11pt]{article}

\usepackage[letterpaper,margin=1in]{geometry}
\usepackage{amsmath, amssymb, amsthm, amsfonts}
\usepackage{bbm}
\usepackage{bm}
\usepackage{complexity}

\usepackage{subcaption}

\usepackage{ifthen}
\usepackage{comment}
\usepackage{tikz}
\usetikzlibrary{positioning,decorations.pathreplacing}
\usepackage{comment}

\usepackage{cite}
\usepackage{appendix}
\usepackage{graphicx}
\usepackage{color}
\usepackage{algorithm}
\usepackage[noend]{algpseudocode}
\usepackage{epstopdf}
\usepackage{wrapfig}
\usepackage{paralist}
\usepackage[textsize=tiny]{todonotes}

\usepackage{framed}
\usepackage[framemethod=tikz]{mdframed}
\usepackage[bottom]{footmisc}
\usepackage{enumitem}
\setitemize{noitemsep,topsep=3pt,parsep=3pt,partopsep=3pt}
\usepackage[font=small]{caption}
\usepackage{xspace}
\definecolor{darkgreen}{rgb}{0,0.5,0}
\usepackage{hyperref}
\hypersetup{
    unicode=false,          
    colorlinks=true,        
    linkcolor=red,          
    citecolor=darkgreen,        
    filecolor=magenta,      
    urlcolor=cyan           
}
\usepackage[capitalize]{cleveref}

\newtheorem{theorem}{Theorem}[section]
\newtheorem{lemma}[theorem]{Lemma}
\newtheorem{meta-theorem}[theorem]{Meta-Theorem}

\newtheorem{remark}[theorem]{Remark}

\newtheorem{proposition}[theorem]{Proposition}
\newtheorem{observation}[theorem]{Observation}
\newtheorem{definition}[theorem]{Definition}

\newcommand{\eps}{\varepsilon}

\let\E\undefined
\DeclareMathOperator{\E}{\mathbb{E}}

\newcommand{\LCS}{\textnormal{LCS}}

\newcommand{\Count}{\textsf{count}}
\newcommand{\bias}{\textsf{bias}}
\newcommand{\adv}{\textsf{adv}}
\newcommand{\Var}{\textnormal{Var}}

\newcommand{\Decode}{\textnormal{Dec}}

\newcommand{\mC}{\mathcal{C}}
\newcommand{\textin}{\textnormal{in}}
\newcommand{\textout}{\textnormal{out}}

\newcommand{\exclude}[1]{}

\newcommand{\FullOrShort}{short}

\ifthenelse{\equal{\FullOrShort}{full}}{

  \newcommand{\fullOnly}[1]{#1}	
  \newcommand{\shortOnly}[1]{}
  }{
    \newcommand{\fullOnly}[1]{}
	\newcommand{\shortOnly}[1]{#1}
  }
  
  \newcommand{\STOConly}[1]{}
\newcommand{\noSTOC}[1]{#1}

\begin{document}

\date{}

\title{Optimally Resilient Codes for List-Decoding\\ from Insertions and Deletions}


\author{Venkatesan Guruswami\footnote{Supported in part by NSF grant CCF-1814603.}
\\Carnegie Mellon University\\ \texttt{venkatg@cs.cmu.edu} \and Bernhard Haeupler\footnote{Supported in part by NSF grants CCF-1527110, CCF-1618280, CCF-1814603, CCF-1910588, NSF CAREER award CCF-1750808 and a Sloan Research Fellowship.}
\\Carnegie Mellon University\\ \texttt{haeupler@cs.cmu.edu} \and Amirbehshad Shahrasbi\footnotemark[2]\\Carnegie Mellon University\\ \texttt{shahrasbi@cs.cmu.edu}}

\maketitle
\thispagestyle{empty}

\begin{abstract}
We give a complete answer to the following basic question: ``What is the maximal fraction of deletions or insertions tolerable by $q$-ary list-decodable codes with non-vanishing information rate?''

\smallskip
This question has been open even for binary codes, including the restriction to the binary insertion-only setting, where the best-known result was that a $\gamma \leq 0.707$ fraction of insertions is tolerable by some binary code family. 

\smallskip
For any desired $\eps > 0$, we construct a family of binary codes of positive rate which can be efficiently list-decoded from any combination of $\gamma$ fraction of insertions and $\delta$ fraction of deletions as long as $ \gamma + 2\delta \leq 1 -\eps$. On the other hand, for any $\gamma, \delta$ with $\gamma + 2\delta = 1$ list-decoding is impossible.
Our result thus precisely characterizes the feasibility region of binary list-decodable codes for insertions and deletions.

\smallskip
We further generalize our result to codes over any finite alphabet of size $q$. Surprisingly, our work reveals that the feasibility region for $q>2$ is \emph{not} the natural generalization of the binary bound above. 
 We provide tight upper and lower bounds that precisely pin down the feasibility region, which turns out to have a $(q-1)$-piece-wise linear boundary whose $q$ corner-points lie on a quadratic curve.

 \smallskip
 The main technical work in our results is proving the existence of code families of sufficiently large \emph{size} with good list-decoding properties for any combination of $\delta,\gamma$ within the claimed feasibility region. We achieve this via an intricate analysis of codes introduced by [Bukh, Ma; SIAM J. Discrete Math; 2014]. Finally, we give a simple yet powerful concatenation scheme for list-decodable insertion-deletion codes which transforms any such (non-efficient) code family (with vanishing information rate) into an efficiently decodable code family with constant rate.

\end{abstract}
	
\newpage
\tableofcontents
\thispagestyle{empty}

\newpage
\setcounter{page}{1}

\section{Introduction}

Error correcting codes have the ability to efficiently correct large fractions of errors while maintaining a large communication rate. The fundamental trade-offs between these two conflicting desiderata have been intensely studied in information and coding theory.
Algorithmic coding theory has further studied what trade-offs can be achieved \emph{efficiently}, i.e., with polynomial time encoding and decoding procedures.

\smallskip

This paper studies insdel codes, i.e., error correcting codes with a large minimum edit distance, which can correct synchronization errors such as insertions and deletions. While codes for Hamming errors and the Hamming metric are quite well understood, insdel codes have largely resisted such progress but have attracted a lot of attention recently~\cite{brakensiek2016efficient,bukh2017improved,guruswami2016efficiently,guruswami2017deletion,haeupler2017synchronization,haeupler2017synchronization3,haeupler2017synchronization2,haeupler2019near,cheng2018synchronization,liu2019list,liu2019explicit,haeupler2018optimal,cheng2019block,cheng2018deterministic}. 
A striking example of a basic question that is open in the context of synchronization errors is the determination of the maximal fraction of deletions or insertions a unique- or list-decodable binary code with non-vanishing rate can tolerate.
%
%
That is, we do not even know at what fraction of errors the rate/distance tradeoff for insdel codes hits zero rate. These basic and intriguing questions are open even if one just asks about the existence of codes, irrespective of computational considerations, and even when restricted to the insertion-only setting.

\smallskip

In this paper we fully answer these questions for list-decodable binary codes and more generally for codes over any alphabet of a fixed size $q$. Our results are efficient and work for any combination of insertions and deletions from which list decoding is information-theoretically feasible at all.

\subsection{Prior Results and Related Works}

The study of codes for insertions and deletions has a long history and goes back to studies of Levenshtein\cite{Levenshtein65} in the 60s. We refer to the surveys by Sloan~\cite{sloane2002single}, Mercier et al.~\cite{mercier2010survey} and Mitzenmacher~\cite{mitzenmacher2009survey} for a more extensive background, and focus here on works related to the main thrust of this paper, namely the maximal tolerable fraction of worst-cast deletions or insertions for unique- and list-decodable code families with non-vanishing rate. We stress that our focus is on \emph{worst-case} patterns of insdel errors subject to bounds on the fraction of insertions and the fraction of deletions allowed. There is also a rich body of work on tackling random insdel errors, which is not the focus of this work.

\smallskip 

\noindent \textbf{Unique Decoding.}
Let us first review the situation for unique decoding, where the decoder must determine the original transmitted codeword.
For unique decoding of binary codes, the maximal tolerable fraction of deletions is easily seen to be at most $\frac{1}{2}$ because otherwise either all zeros or all ones in a transmitted codeword can be deleted. (For $q$-ary codes, this fraction becomes $1-1/q$.) On the other hand, for a long time the best (existential) possibility results for unique-decodable binary codes stemmed from analyzing random binary codes. 

In the Hamming setting, random codes often achieve the best known parameters and trade-offs, and a lot of effort then goes into finding efficient constructions and decoding algorithms for codes that attempt to come close to the random constructions. However, the edit distance is combinatorially intricate and even analyzing the expected edit distance of two random strings, which is the first step in analyzing random codes, is highly non-trivial.

 %
 %
 Lueker~\cite{lueker2009improved}, improving upon earlier results by Dan{\v{c}}{\'\i}k and Paterson~\cite{dancik1994expected,danvcik1995upper}, proved that the expected fractional length of the longest common subsequence between two random strings lies between 0.788071 and 0.826280 (the exact value is still unknown). 
 Using this, one can show that a random binary code of positive rate can tolerate between $0.23$ and $0.18$ fraction of deletions or insertions. 
 Edit distance of random $q$-ary strings were studied by Kiwi, Loebl, and  Matou\~sek\cite{kiwi2005expected}, leading to positive rate random codes by Guruswami and Wang \cite{guruswami2017deletion} that correct $1-\Theta(\frac{1}{\sqrt{q}})$ fraction of deletions for asymptotically large $q$. Because random codes do not have efficient decoding and encoding procedures these results were purely existential. Computationally efficient binary codes of non-vanishing rate tolerating some small unspecified constant fraction of insertions and deletions were given by Schulman and Zuckerman~\cite{schulman1999asymptotically}. Guruswami and Wang~\cite{guruswami2017deletion} gave binary codes that could correct a small constant fraction of deletions with rate approaching $1$, and this was later extended to handle insertions as well~\cite{guruswami2016efficiently}.
 
 In the regime of low-rate and large fraction of deletions, Bukh and Guruswami~\cite{bukh2016improved} gave a $q$-ary code construction that could tolerate up to a $\frac{q-1}{q+1}$ fraction of deletions, which is $\frac{1}{3}$ for binary codes. Note that this beats the performance of random codes. Together with H\aa stad~\cite{bukh2017improved} they later improved the deletion fraction to $1 - \frac{2}{q+\sqrt{q}}$ or $\sqrt{2}-1 \approx 0.414$ for binary codes. This remains the best known result for unique-decodable codes and determining whether there exist binary codes capable of correcting a fraction of deletions approaching $\frac{1}{2}$ remains a fascinating open question.

\smallskip 
\noindent \textbf{List decoding.}
The situation for list-decodable codes over small alphabets is equally intriguing. In list-decoding, one relaxes the decoding requirement from having to output the codeword that was sent to having to produce a (polynomially) small list of codewords which includes the correct one. 
The trivial limit of $1/2$ fraction deletions for unique-decoding binary codes applies equally well for list-decoding. In their paper, Guruswami and Wang~\cite{guruswami2017deletion} showed that this limit can be approached by efficiently list-decodable binary codes. Similarly, $q$-ary codes list-decodable from a deletion fraction approaching the optimal $1-1/q$ bound can be constructed.

However, the situation was not well understood when insertions are also allowed. 
It had already been observed by Levenshtein~\cite{Levenshtein65} that (at least existentially) insertions and deletions are equally hard to correct for unique-decoding, in that if a code can correct $t$ deletions then it can also correct any combination of $t$ insertions and deletions. This turns out to be not true for list-decoding. This was demonstrated pointedly in \cite{haeupler2018synchronization4}, where it is shown that arbitrary large $\gamma=O(1)$ fractions of insertions (possibly exceeding 1) can be tolerated by list-decodable codes over sufficiently large constant alphabets (see \cref{thm:LargeAlphaInsDelListDecoding}), whereas the fraction of deletions $\delta$ is clearly bounded by $1$. Indeed, the fraction of insertions $\gamma$ does not even factor into the rate of these list-decodable insertion-deletion codes---this rate can approach the optimal bound of $1-\delta$ where $\delta$ is the deletion fraction. The result in \cite{haeupler2018synchronization4}, however, applies only to sufficiently large constant alphabet sizes, and it does not shed any light on the list-decodability of \emph{binary} (or any fixed alphabet) insdel codes. 

Considering a combination of insertions and deletions, the following bound is not hard to establish.
\begin{proposition}\label{thm:impossibility}
For any integer $q$ and any $\delta,\gamma \geq 0$ with $\frac{\delta}{1-\frac{1}{q}} + \frac{\gamma}{q-1}  \geq 1$ there is no family of constant rate codes of length $n$ which are list-decodable from $\delta n$ deletions and $\gamma n$ insertions. 
\end{proposition}


For the case of insertion-only binary codes, the above limits the maximum fraction of insertions to $100\%$, which is twice as large as the best possible deletion fraction of $1/2$. 

Turning to existence/constructions of list-decodable codes for insertions, recall that the codes of Bukh, Guruswami, H\aa stad (BGH) could unique-decode (and thus also list-decode) a fraction of $0.414$ insertions (indeed any combination of insertions and deletions totaling $0.414$ fraction). Wachter-Zeh~\cite{wachter2017list} recently put forward a Johnson-type bound for insdel codes. The classical Johnson bound works in the Hamming metric, and connects unique-decoding to list-decoding (for Hamming errors) by showing that any unique-decodable code must also be list-decodable from an even larger fraction of corruptions. One intriguing implication of Wachter-Zeh's Johnson bound for insdel codes is that any unique-decodable insdel code which tolerates a $\frac{1}{2}$ fraction of deletions (or insertions) would automatically also have to be (existentially) list-decodable from a $100\%$ fraction of insertions. 
Therefore, even if one is interested in unique-decoding, e.g., closing the above-mentioned gap between $\sqrt{2}-1$ and $\frac{1}{2}$, this establishes the search for maximally list-decodable binary codes from insertions as a good and indeed necessary step towards this goal. On the other hand, proving any non-trivial impossibility result bounding the maximal fraction of insertions of list-decodable binary codes away from $100\%$ would directly imply an impossibility result for unique-decoding binary codes from a deletion fraction approaching $\frac{1}{2}$. 

Follow-up work by Hayashi and Yasunaga~\cite{hayashi2018list} corrected some subtle but crucial bugs in \cite{wachter2017list} and reproved a corrected Johnson Bound for insdel codes. They furthermore showed that the BGH codes~\cite{bukh2017improved} could be list-decoded from a fraction $\approx 0.707$ of insertions. Lastly, via a concatenation scheme used in \cite{guruswami2017deletion,guruswami2016efficiently} they furthermore made these codes efficient. A recent work of Liu, Tjuawinata, and Xing~\cite{liu2019list} also provides efficiently list-decodable insertion-deletion codes and derives a Zyablov-type bound.
In summary, for the binary insertion-only setting, the largest fraction of insertions that we knew to be list-decodable (even non-constructively) was $\approx 0.707$.


\subsection{Our Results}

We close the above gap and show binary codes which can be list-decoded from a fraction $1-\eps$ fraction of insertions, for any desired constant $\eps > 0$. In fact, we give a single family of codes that are list-decodable from any mixed combination of $\gamma$ fraction of insertions and $\delta$ fraction of deletions, as long as $2 \delta + \gamma \le 1 - \eps$. 
\begin{theorem}\label{thm:main}
For any $\eps\in(0, 1)$ and sufficiently large $n$, there exists a constant rate family of efficient binary codes that are $L$-list decodable from any $\delta n$ deletions and $\gamma n$ insertions in $\poly(n)$ time as long as $\gamma + 2\delta \le 1-\eps$ where $n$ denotes the block length of the code, $L=O_\eps(\exp(\exp(\exp(\log^*n))))$, and the code achieves a rate of $\exp\left(-\frac{1}{\eps^{10}}\log^2 \frac{1}{\eps}\right)$.
\end{theorem}
Since the computationally efficient codes from \Cref{thm:main} match the bounds from \Cref{thm:impossibility} for every $\delta,\gamma$,  this nails down the entire feasibility region for list-decodability from insertions and deletions for the binary case. We stress that while we get constructive results, even the existence of inefficiently list-decodable codes, that too just for the insertion-only setting, was not known prior to this work.

In the above result, the rather weird looking bound on the list-size is inherited from results on list-decoding from a huge number insertions over larger alphabets~\cite{haeupler2018synchronization4}, which in turn is inherited from the list-size bounds for the list-recoverable algebraic-geometric code constructions in \cite{guruswami2013list}.

\smallskip

We use similar construction techniques to obtain codes with positive rate over any arbitrary alphabet size $q$ that are list-decodable from any fraction of insertions and deletions under which list-decoding is possible. We thus precisely identify the feasibility region for any alphabet size, together with an efficient construction. Again, recall that the existence of such codes was not known earlier, even for the insertion-only case.

\begin{theorem}\label{thm:qaryMain}
For any positive integer $q \geq 2$, define $F_q$ as the concave polygon defined over vertices
$\left(\frac{i(i-1)}{q}, \frac{q-i}{q}\right)$ for $i = 1, \cdots, q$ and $(0, 0)$. (An illustration for $q=5$ is presented in \cref{fig:actual-region}). $F_q$ does not include the border except the two segments 
$\left[(0, 0), (q-1, 0)\right)$ and $\left[(0, 0), \left(0, 1-1/q\right)\right)$. 
Then, for any $\eps > 0$ and sufficiently large $n$, there exists a family of $q$-ary codes that, as long as $(\gamma, \delta) \in (1-\eps)F_q$, are efficiently $L$-list decodable from any $\delta n$ deletions and $\gamma n$ insertions where $n$ denotes the block length of the code, $L=O(\exp(\exp(\exp(\log^*n))))$, and the code achieves a positive rate of $\exp\left(-\frac{1}{\eps^{10}}\log^2 \frac{1}{\eps}\right)$.
\end{theorem}
We further show in \cref{sec:large-alphabet} that for any pair of positive real numbers $(\gamma, \delta) \not\in F_q$, there exists no infinite family of $q$-ary codes with rate bounded away from zero that can be list decoded from a $\delta$-fraction of deletions plus a $\gamma$-fraction of insertions.

\begin{figure}[]
  \centering
  \noSTOC{\includegraphics[height=2in]{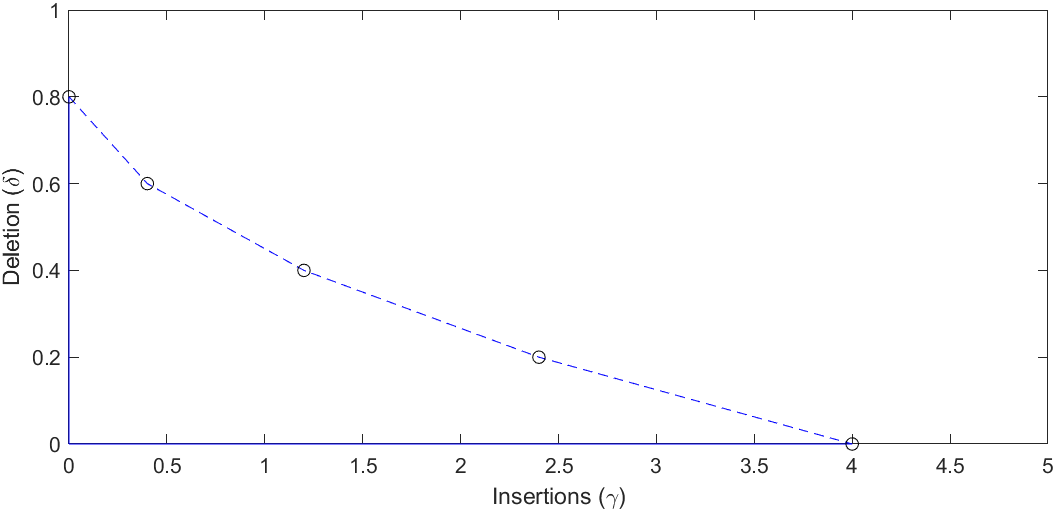}}
  \STOConly{\includegraphics[width=\linewidth]{Region.png}}
  \caption{Feasibility region for $q=5$.}\label{fig:actual-region}
\end{figure}

\subsection{Our Techniques}

We achieve these results using two ingredients, each interesting in its own right. The first is a simple new concatenation scheme for list-decodable insdel codes which can be used to boost the rate of insdel codes. The second component, which constitutes the bulk of this work, is a technically intricate proof of the list-decoding properties of the Bukh-Ma codes~\cite{bukh2014longest} which have good (edit) distance properties but a tiny sub-constant rate. We note that these codes were the inner codes in the ``clean construction" in the BGH work on codes unique-decodable from a $1/3$ insdel fraction~\cite{bukh2017improved}. This was driven by a property of these codes called the \emph{span}, which is a stronger form of edit distance that applies at all scales. The Bukh-Ma codes were also used by Guruswami and Li~\cite{GL-oblivious-deletion} in their existence proof of codes of positive rate for correcting a fraction of \emph{oblivious deletions} approaching 1. In this work, the non-trivial list-decodability property of the Bukh-Ma codes drives our result.

\subsubsection{Concatenating List-Decodable Insdel Codes}

Our first ingredient is a simple but powerful framework for constructing list-decodable insertion-deletion codes via code concatenation. Recall that code concatenation which composes the encoding of an \emph{outer} code $C_{\rm out}$ with an \emph{inner} code $C_{\rm in}$ whose size equals the alphabet size of $C_{\rm out}$.

In our approach, the outer code $C_{\rm out}$ is chosen to be 
a list-decodable insdel code $C_{\rm out}$ over an alphabet that is some large function of $1/\eps$, but which has constant rate and is  capable of tolerating a huge number of insertions.
The inner code $C_{\rm in}$ is chosen to be a list-decodable insdel code over a fixed alphabet of the desired size $q$, which has non-trivial list decoding properties for the desired fraction $\delta,\gamma$ of deletions and insertions.

We show that even if $C_{\rm in}$ has an essentially arbitrarily bad sub-constant rate and is not efficient, the resulting $q$-ary insdel code does have constant rate, and can also be efficiently list decoded  from the same fraction of insertions and deletions as $C_{\rm in}$. For the problem considered in this paper, this framework essentially provides efficiency of codes for free. More importantly, it reduces the problem of finding good \emph{constant-rate} insdel codes over a fixed alphabet to finding a family of good list-decodable insdel codes \emph{with an arbitrarily large number of codewords}, and a list-size bounded by some fixed function of $1/\eps$. 

\smallskip

Our decoding procedure for concatenated list-decodable insdel codes is considerably simpler than similar schemes introduced in earlier works~\cite{guruswami2017deletion,guruswami2016efficiently,bukh2017improved,schulman1999asymptotically}. Of course, the encoding is simply given by the standard concatenation procedure. The decoding is done by (i) list-decoding shifted intervals of the received string using the inner code $C_{\rm in}$, (ii) creating a single string from the symbols in these lists, and (iii) using the list-decoding algorithm of the outer code on this string (viewed as a version of the outer codeword with some number of deletions and insertions). 

The main driving force behind why this simplistic sounding approach actually works is a judicious choice of the outer code $C_{\rm out}$. Specifically, we use the codes due to Haeupler, Shahrasbi, and Sudan~\cite{haeupler2018synchronization4} which can tolerate a very large number of insertions. This means that the many extra symbols coming from the list-decodings of the inner code $C_{\rm in}$ and the choice of overlapping intervals does not disrupt the decoding of the outer code. 

\subsection{\STOConly{\hspace{-1mm}}Analyzing the \noSTOC{List-Decoding }Properties of Bukh-Ma Codes}

The main technical challenge that remains is to construct or prove the existence of arbitrarily large binary codes with optimal list decoding properties for any $\gamma, \delta$ (and $q$). For this we turn to a simple family of codes introduced by Bukh and Ma~\cite{bukh2014longest}, which consist of strings $(0^r\ 1^r)^{\frac{n}{r}}$ which oscillate between $0$'s and $1$'s with different frequencies. (Below we will refer to $r$ as the \emph{period}, and $1/r$ should be thought of as the \emph{frequency} of alternation.) 

A simple argument shows that the edit distance between any two such strings with sufficiently different periods is maximal, resulting in a tolerable fraction of edit errors of $\frac{1}{2}$ for unique decoding. The Johnson bound of \cite{wachter2017list,hayashi2018list} implies that this code must also be list-decodable from a full fraction $100\%$ of insertions. Therefore, using these codes as the inner codes in the above-mentioned concatenation scheme resolves the list-decoding question for the insertion-only setting. (The deletion-only setting is oddly easier as just random inner codes suffice, and was already resolved in \cite{guruswami2017deletion}.) This also raises hope that the Bukh-Ma codes might have good list-decoding properties for other $\gamma,\delta$ as well. Fortunately, this turns out to be true, though establishing this involves an intricate analysis that constitutes the bulk of the technical work in this paper.

\begin{theorem}\label{thm:LowRateListDecBinaryCodes}
For any $\eps>0$ and sufficiently large $n$, let $C_{n, \eps}$ be the following Bukh-Ma code:
$${C}_{n, \eps} = \left\{\left(0^r1^r\right)^{\frac{n}{2r}}\Big| r=\left(\frac{1}{\eps^4}\right)^k, k< \log_{1/\eps^4} n\right\}.$$

For any $\delta,\gamma \geq 0$ where $\gamma+2\delta < 1-\eps$, ${C}_{n, \eps}$ is list-decodable from any $\delta n$ deletions and $\gamma n$ insertions with\noSTOC{ a list size of $O\left(\frac{1}{\eps^3}\right)$.}\STOConly{ $O(\eps^{-3})$ list size.}
\end{theorem}

In order to prove \cref{thm:LowRateListDecBinaryCodes} we first introduce a new correlation measure which expresses how close a string is to any given frequency (or Bukh-Ma codeword) if one allows for both insertions and deletions each weighted appropriately. Using this we want to show that it is impossible to have a single string $v$ which is more than $\eps$-correlated with more than $\Theta_\eps(1)$ frequencies. 

Intuitively, one might expect that each correlation can be (fractionally) attributed to a (disjoint) part of $v$ which would result in the maximum number of $\eps$-close frequencies to be at most $1/\eps$. This, however, turned out to be false. Instead, we use a proof technique which is somewhat reminiscent of the one used to establish the polarization of the martingale of entropies in the analysis of polar codes~\cite{arikan2008channel,blasiok2018polar}. 

In more detail, we think of recursively sub-sampling smaller and smaller nested substrings of $v$, and analyze the expectation and variance of the bias between the fraction of $0$'s and $1$'s in these substrings. More precisely, we order the run lengths $r_1,r_2,\ldots$ that are $\eps$-correlated with $v$ in decreasing order and first sample a substring $v_1$ with $r_1 \gg |v_1| \gg r_2$ from $v$. While the expected zero-one bias in $v_1$ is the same as in $v$, we show that the variance of this bias is an increasing function in the correlation with $\left(0^{r_1} 1^{r_1}\right)^{\frac{n}{2{r_1}}}$. Intuitively, $v_1$ cannot be too uniform on an scale of length $l$ if it is correlated with $r_1$. 

Put differently, in expectation the sampled substring $v_1$ will land in a part of $v$ which is either (slightly) correlated to one of the long stretches of zeros in $v$ or in a part which is correlated with a long stretch of ones in $v$, resulting in at least some variance in the bias of $v_1$. Because the scales $r_2, r_3, \ldots$ are so much smaller than $v_1$, this sub-sampling of $v_1$ furthermore preserves the correlation with these scales intact, at least in expectation.

Next we sample a substring $v_2$ with $r_2 \gg |v_2| \gg r_3$ within $v_1$. Again, the bias in $v_2$ stays the same as the one in $v_1$ in expectation but the sub-sampling introduces even more variance given that $v_1$ is still non-trivially correlated with the string with period $r_2$. The evolution of the bias of the strings $v_1, v_2, \ldots$ produced by this nested sampling procedure can now be seen as a martingale with the same expectation but an ever increasing variance. Given that the bias is bounded in magnitude by 1, the increase in variance cannot continue indefinitely. This limits the number of frequencies a string $v$ can be non-trivially correlated with, which is exactly what we were after.

Our generalization to larger $q$-ary alphabets follows the same high level blueprint, but is technically even more delicate. Recall that in the non-binary case, there are $(q-1)$ different linear trade-offs between $\delta,\gamma$ depending on the exact regime they lie in.

\section{Preliminaries}
\subsection{List-Decodable Insertion-Deletion Codes}
The following list-decodable insertion-deletion codes from \cite{haeupler2018synchronization4} will be used as the outer code in our constructions.
\begin{theorem}[Theorem 1.1 from~\cite{haeupler2018synchronization4}]\label{thm:LargeAlphaInsDelListDecoding}
For every $\delta,\eps\in(0, 1)$ and constant $\gamma > 0$, there exist a family of list-decodable insdel codes that can protect against $\delta$-fraction of deletions and $\gamma$-fraction of insertions and achieves a rate of $1-\delta-\eps$ or more over an alphabet of size 
$\left(\frac{\gamma+1}{\eps^2}\right)^{O\left(\frac{\gamma+1}{\eps^3}\right)}=O_{\gamma, \eps}\left(1\right)$. These codes are list-decodable with lists of size $L_{\eps, \gamma}(n)= \exp\left(\exp\left(\exp\left(\log^* n\right)\right)\right)$, and have polynomial time encoding and decoding complexities.
\end{theorem}

\subsection{Strings, Insertions\noSTOC{ and}\STOConly{,} Deletions, and Distances}

In this section we provide preliminary definitions on strings, edit operations, and related notions. \noSTOC{We start by definition of count and bias.}

\begin{definition}[Count and Bias]
We define $\Count_a(w)= |\{i | w[i]=a\}|$ as the number of appearances of symbol $a$ in string $w$. The bias of a binary string $w$ is the normalized difference between the appearances of zeros and ones in $w$, i.e., $\bias(w) = \frac{\Count_1(w)-\Count_0(w)}{|w|}$. With this definition, 
$\Count_0(w) = \frac{1-\bias(w)}{2}|w|$ and $\Count_1(w) = \frac{1+\bias(w)}{2}|w|$.
\end{definition}

\noSTOC{Next, we formally define a \emph{matching} between two strings. }
\begin{definition}[Matching]
A matching $M$ of size $k$ between two strings $S$ and $S'$ is defined to be two sequences of $k$ integer positions $0 < i_1 < \ldots < i_k \leq |S|$ and $0 < i'_1 < \ldots < i'_k \leq |S'|$ for which $S[i_j]=S'[i'_j]$ for all $j \leq k$. The subsequence induced by a matching $M$ is simply $S[i_1],\ldots,S[i_k]$. 
Every common subsequence between $S$ and $S'$ implicitly corresponds to a matching and we use the two interchangeably.
\end{definition}

\noSTOC{We now proceed to define the important notion of advantage.}

\begin{definition}[Advantage of a Matching]
Let $M$ be a matching between two binary strings $a$ and $b$. The \emph{advantage of the matching $M$} is defined as 
\noSTOC{$$\adv_M = \frac{3|M|-|a|-|b|}{|a|}.$$}\STOConly{$\adv_M = \frac{3|M|-|a|-|b|}{|a|}.$}
\end{definition}

\begin{definition}[Advantage]
For a given pair of strings $a$ and $b$, the \emph{advantage of $a$ to $b$} is defined as the advantage of the matching $M$ that corresponds to the largest common subsequence between them, i.e.,
$\adv(a, b) = \adv_{M=\LCS(a, b)}$. It is easy to verify that the longest common subsequence $M$ maximizes the advantage among all matchings from $a$ to $b$.
\end{definition}

We now make the following remark that justifies the notion of advantage as defined above. Note that any matching between two strings $a$ and $b$ implies a set of insertions and deletions to convert $b$ to $a$ which is, to delete all unmatched symbols in $b$ and insert all unmatched symbols in $a$ within the remaining symbols.

\begin{remark}\label{rmk:advInterpretation}
Consider strings $a$ and $b$ and matching $M$ between them. Think of $a$ as a distorted version of $b$ and let $\delta_M$ and $\gamma_M$ represent the fraction of deletions and insertions needed to convert $b$ to $a$ as suggested by $M$, i.e.,
\noSTOC{$$\delta_M  = \frac{\textnormal{Number of unmatched symbols in $b$}}{|b|}=\frac{|b|-|M|}{|b|},$$
and 
$$\gamma_M = \frac{\textnormal{Number of unmatched symbols in $a$}}{|b|}=\frac{|a|-|M|}{|b|}.$$}\STOConly{$\delta_M  = \frac{\textnormal{Number of unmatched symbols in $b$}}{|b|}=\frac{|b|-|M|}{|b|},$ and $\gamma_M = \frac{\textnormal{Number of unmatched symbols in $a$}}{|b|}=\frac{|a|-|M|}{|b|}.$}
The $\adv_M$ function tracks the value of $|b|(1-2\delta_M-\gamma_M)$ normalized by $|a|$ rather than $|b|$.
\noSTOC{$$\adv_M(a, b) = \frac{3|M|-|a| - |b|}{|a|} = \frac{3|b|(1-\delta_M) - |b|(1-\delta_M+\gamma_M) - |b|}{|a|}=\frac{|b|}{|a|}\cdot (1- 2\delta_M-\gamma_M)$$}\STOConly{
\begin{align*}
&\adv_M(a, b) = \frac{3|M|-|a| - |b|}{|a|}\\
&= \frac{3|b|(1-\delta_M) - |b|(1-\delta_M+\gamma_M) - |b|}{|a|}=\frac{|b|}{|a|}\cdot (1- 2\delta_M-\gamma_M)
\end{align*}
}
We will make use of this unnatural normalization later on. 
\end{remark}

We now extend the definition of advantage to the case where the second argument is an infinite string.

\begin{definition}[Infinite Advantage]
For a finite string $a$ and infinite string $b$, the advantage of $a$ to $b$ is defined as the minimum advantage that $a$ has over all substrings of $b$.
$$\adv(a, b) = \min_{b' = b[i, j]}\adv(a, b').$$
\end{definition}

We now define a family of binary strings called \emph{Alternating Strings}.

\begin{definition}[Alternating Strings]
For any positive integer $r$, we define the infinite alternating string of run-length $r$ as $A_r = (0^r1^r)^\infty$ and denote its prefix of length $l$ with $A_{r,l} = A_r[1,l]$.
\end{definition}

We finish the preliminaries by the following lemma stating some properties of the notions defined through this section.

\begin{lemma}\label{lem:advProperties}
The following properties hold true:
\begin{itemize}
    \item For any pair of binary strings $S_1, S_2$ where $\adv(S_1, S_2) > 0$, lengths of $S_1$ and $S_2$ are within a factor of two of each other, i.e, $\min(|S_1|, |S_2|) \geq \frac{\max(|S_1|, |S_2|)}{2}$.
    \item For any binary string $S$ and integer $r$, $\adv(S, A_{r}) \geq -\frac{1}{2}$
\end{itemize}
\end{lemma}
\begin{proof}
For the first part, let $M=\LCS(S_1, S_2)$. We have that 
    $\adv(S_1, S_2) \geq 0\Rightarrow 3|M|\geq|S_1|+|S_2|$,
    which, as $|M|\leq\min(|S_1|, |S_2|)$, implies that $\min(|S_1|, |S_2|) \geq \frac{\max(|S_1|, |S_2|)}{2}$.
    
    For the second part, let $n=|S|$ and assume that $b\in\{0, 1\}$ is the most frequent bit in $S$ and there are $m$ occurrences of $b$ in $S$. Take a substring $S'$ in $A_r$ as the smallest string that starts at the beginning of a $b^r$ block and contains the same number of $b$s as $S$. The size of $S'$ is no more than $2m$ and the longest common subsequence between $S$ and $S'$ is at least $m$. Therefore,
\noSTOC{\begin{equation*}
        \adv(S, A_{r})\geq \adv(S, S')\geq\frac{3|M|-|S|-|S'|}{|S|}\geq\frac{3m-2m-2m}{n} \geq \frac{-m}{n} \geq -\frac{1}{2}.\qedhere
    \end{equation*}}\STOConly{$\adv(S, A_{r})\geq \adv(S, S')\geq\frac{3|M|-|S|-|S'|}{|S|}\geq\frac{3m-2m-2m}{n} \geq \frac{-m}{n} \geq -\frac{1}{2}.$\qedhere}
\end{proof}

\section{Proof of \STOConly{Theorem \ref{thm:LowRateListDecBinaryCodes}}\noSTOC{\cref{thm:LowRateListDecBinaryCodes}}: List-Decoding for Bukh-Ma Codes}

To prove this theorem, we assume for the sake of contradiction that there exists a string $v$ and $k > \frac{1200}{\eps^3}$ members of ${C}_{n, \eps}$ like $A_{r_1, n}, A_{r_2, n}, \cdots, A_{r_k, n}$, so that each $A_{r_i, n}$ can be converted to $v$ with $I_i$ insertions and $D_i$ deletions where $I_i+2D_i \leq n(1-\eps)$. We define the indices in a way that $r_1 > r_2 > \cdots > r_k$. Given the definition of ${C}_{n, \eps}$, $r_{i} \geq \frac{r_{i+1}}{\eps^4}$. 
We first show that, for \noSTOC{all }$i=1, 2, \cdots, k$,  $\adv(v, A_{r_i, n}) \geq \frac{\eps}{2}$.

\begin{lemma}\label{lem:advantage}
For any $1\leq i\leq k$, $\adv(v, A_{r_i, n}) \geq \frac{\eps}{2}$.
\end{lemma}
\begin{proof}
Let $M_i$ denotes the matching that corresponds to the set of $I_i$ insertions and $D_i$ deletions that convert $A_{r_i, n}$ to $v$. 
$$I_i+2D_i \leq n(1-\eps) \Rightarrow n-I_i-2D_i \geq n\eps\Rightarrow 1-\gamma_i-2\delta_i \geq \eps
$$
Note that according to \cref{rmk:advInterpretation}, 
$\adv(v, A_{r_i, n}) = \frac{n}{|v|}\cdot(1-\gamma_i-2\delta_i)$. Thus,
$\adv(v, A_{r_i, n}) \geq \frac{n}{|v|} \eps \geq \frac{\eps}{2}$. 
The last step follows from the first item of \cref{lem:advProperties}.
\end{proof}

Having \cref{lem:advantage}, we are ready to prove \cref{thm:LowRateListDecBinaryCodes}. We start with defining a couple of sequences of random variables via random sampling of nested substrings of $v$.
We split the string $v$ into substrings of size $l_1 = r_1\eps^2$, pick one uniformly at random and denote it by $v_1$. We define random variable $A_1 = \adv(v_1, A_{r_1})$ and random variable $B_1 = \bias(v_1)$. Similarly, we split $v_1$ into substrings of length $l_2 = r_2\eps^2$ and pick $v_2$ uniformly at random and define $A_2 = \adv(v_2, A_{r_2})$ and $B_2 = \bias(v_2)$. Continuing this procedure, one can obtain the two sequences of random variables $A_1, A_2, \cdots, A_k$ and $B_1, B_2, \cdots, B_k$. We will prove the following.
\noSTOC{
\begin{lemma}\label{lem:process-guarantees}
The following hold for $A_1, A_2, \cdots, A_k$ and $B_1, B_2, \cdots, B_k$.
\begin{enumerate}
\item $\E[B_i] = \bias(v)$
\item $\E[A_i] \geq \frac{\eps}{2}$
\end{enumerate}
\end{lemma}}
\STOConly{
\begin{lemma}\label{lem:process-guarantees}
The following hold for $A_1, \cdots, A_k$ and $B_1, \cdots, B_k$:\\
$\quad (1) \E[B_i] = \bias(v),\quad (2) \E[A_i] \geq \frac{\eps}{2}.$
\end{lemma}}
\begin{proof}
Note that one can think of $v_i$ as a substring of $v$ that is obtained by splitting $v$ into substrings of length $l_i$ and choosing one uniformly at random. Let $U$ denote the set of all such substrings. We have that
\begin{eqnarray*}
\E[B_i] &=& \sum_{\hat{v}\in U}\frac{1}{|U|}\cdot\bias(\hat{v})
= \frac{1}{|U|}\sum_{\hat{v}\in U}\frac{\Count_1(\hat{v})-\Count_0(\hat{v})}{l_i} \STOConly{\\&}=\STOConly{&} \frac{\Count_1(v)-\Count_0(v)}{|U|\cdot l_i} = \bias(v).
\end{eqnarray*}

A similar argument proves the second item. Take the matching $M_i$ between $v$ and $A_{r_i, n}$ that achieves the advantage $\adv(v, A_{r_i, n})$, i.e., the largest matching between $v$ and $A_{r_i, n}$. Take some $\hat{v}\in U$; $\hat{v}$ is mapped to some substring in $A_{r_i, n}$ under $M_i$. We call that substring of $\hat{v}$, \emph{the projection of $\hat{v}$ under $M_i$} and denote it by $\hat{v} \rightarrow M_i$. We also represent the subset of $M_i$ that appears between $\hat{v}$ and $\hat{v} \rightarrow M_i$ with $M_i[\hat{v}]$.

For a $\hat{v}\in U$, we define $a(\hat{v})$ as the value for advantage that is yielded by the matching $M_i[\hat{v}]$ between $\hat{v}$ and $\hat{v} \rightarrow M_i$. In other words, $a(\hat{v}) = \frac{3|M_i[\hat{v}]| - |\hat{v}|- |\hat{v} \rightarrow M_i|}{|\hat{v}|}$. Given the definitions of advantage and infinite advantage, we have that
\noSTOC{$$a(\hat{v}) \leq \adv(\hat{v}, \hat{v} \rightarrow M_i) \leq \adv(\hat{v}, A_{r_i}).$$}\STOConly{$a(\hat{v}) \leq \adv(\hat{v}, \hat{v} \rightarrow M_i) \leq \adv(\hat{v}, A_{r_i}).$}
This can be used to prove the second item as follows:
\begin{eqnarray*}
\E[A_i] &=& \sum_{\hat{v}\in U} \frac{1}{|U|}\cdot\adv(\hat{v}, A_{r_i}) \geq \frac{1}{|U|}\cdot\sum_{\hat{v}\in U} a(\hat{v})\\
&=&\frac{1}{|U|}\cdot\sum_{\hat{v}\in U} \frac{3|M_i[\hat{v}]|-|\hat{v}|-|\hat{v}\rightarrow M_i|}{|\hat{v}|} \STOConly{\\&}=\STOConly{&}\frac{1}{|U|\cdot|\hat{v}|}\cdot\sum_{\hat{v}\in U} \left(3|M_i[\hat{v}]|-|\hat{v}|-|\hat{v}\rightarrow M_i|\right)\\
&=&\frac{1}{|v|}\cdot \left(3|M_i|-|v|-|A_{r_i, n}|\right)
=\adv(v, A_{r_i, n})\geq \frac{\eps}{2}
\end{eqnarray*}
where the last step follows from \cref{lem:advantage}.
\end{proof}

\begin{lemma}\label{lem:process-variance-guarantee}
For the sequence $B_1, B_2, \cdots, B_k$, we have
\noSTOC{$$\Var(B_{i+1}) \geq \Var(B_i) + \frac{\eps^3}{1200}, \quad \forall 1\leq i < k.$$}
\STOConly{$\Var(B_{i+1}) \geq \Var(B_i) + \frac{\eps^3}{1200}, \quad \forall 1\leq i < k.$}
\end{lemma}
\begin{proof}
To analyze the relation of $\Var(B_i)$ and $\Var(B_{i+1})$, we use the law of total variance and condition the variance of $B_{i+1}$ on $v_i$, i.e., the substring chosen in the $i$th step of the stochastic process, from which we sub sample $v_{i+1}$.
\begin{eqnarray}
\Var(B_{i+1})&=&\Var\left(\E[B_{i+1}|v_{i}]\right) + \E\left[\Var(B_{i+1}|v_{i})\right]\nonumber\\
&=&\Var\left(B_{i}\right) + \E\left[\Var(B_{i+1}|v_{i})\right]\label{eqn:totalVar1}
\end{eqnarray}
Equation \eqref{eqn:totalVar1} comes from the fact that the average bias of substrings of length $l_{i+1}$ in $v_i$ is equal to the bias of $v_i$. Having this, we see that it suffices to show that $\E\left[\Var(B_{i+1}|v_{i})\right] \geq \eps^3/1200$. We remind the reader that $v_{i+1}$ is obtained by splitting $v_i$ into substrings of length $l_{i+1} = r_{i+1}\eps^2$ and choosing one at random. We denote the set of such substrings by $U$. Also, there is a matching $M_i$ between $v_i$ and $A_{r_{i+1}}$ with advantage $\eps$ or more. Any substring of length $l_{i+1}$ is mapped to some substring in $A_{r_{i+1}}$, i.e., its projection of the substring under $M_i$. Note there are three different possibilities for such projection. It is either an all zeros string, an all one string, or a string that contains both zeros and ones. We partition $U$ into three sets $U_0$, $U_1$, and $U_e$ based on which case the projection belongs to. (See \cref{fig:substringTypeAlternting})

\begin{figure}[]
  \centering
  \noSTOC{\includegraphics[width=.75\textwidth]{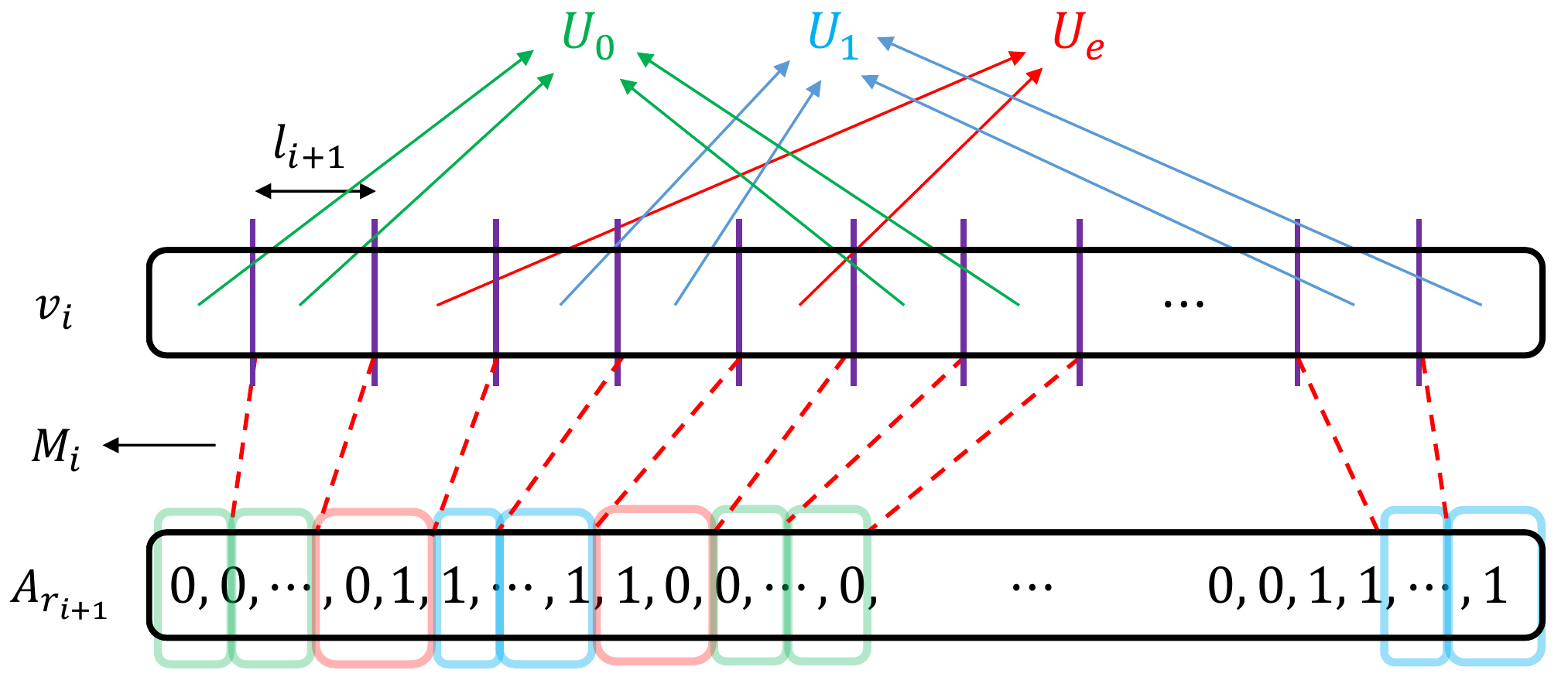}}
  \STOConly{\includegraphics[width=\linewidth]{substringTypeAlternting.pdf}}
  \caption{Partitioning substrings of length $l_{i+1}$ into three sets $U_0, U_1, U_e$}\label{fig:substringTypeAlternting}
\end{figure}

We partition the sample space into three events $E_0$, $E_1$, and $E_e$ based on whether $v_{i+1}$ belongs to $U_0$, $U_1$, or $U_e$ respectively. We also define the random variable $T$ over $\{0, 1, e\}$ that indicates which one of $E_0$, $E_1$, or $E_e$ happens. Once again, we use the law of total variance to bound $\E\left[\Var(B_{i+1}|v_i)\right]$.
\noSTOC{
\begin{eqnarray}
\E\left[\Var(B_{i+1}|v_{i})\right] &=&\E_{v_i}\big[\Var_{T}(\E\left[B_{i+1}|v_i, T\right]) + \E_{T}\left[\Var(B_{i+1}|v_i, T)\right]\big]\nonumber\\
&\geq& \E_{v_i}\big[\Var_{T}(\E\left[B_{i+1}|v_i, T\right])\big]\label{eqn:varLoweboundStep2}
\end{eqnarray}}
\STOConly{
\begin{eqnarray}
\E\left[\Var(B_{i+1}|v_{i})\right] &=&\E_{v_i}\big[\Var_{T}(\E\left[B_{i+1}|v_i, T\right])\nonumber\\ &&+ \E_{T}\left[\Var(B_{i+1}|v_i, T)\right]\big]\nonumber\\
&\geq& \E_{v_i}\big[\Var_{T}(\E\left[B_{i+1}|v_i, T\right])\big]\label{eqn:varLoweboundStep2}
\end{eqnarray}}
Note that the term $\Var_{T}(\E\left[B_{i+1}|v_i, T\right])$ refers to variance of a 3-valued random variable that takes the value $\E_{v_i}\left[B_{i+1}|v_i, T=t\right]$ with probability $\Pr\{T=t|v_i\}$ for $t\in\{0, 1, e\}$. We use three important facts about this distribution to bound its variance from below. 

First, $\Pr\{T=e|v_i\} \leq 2\eps^2$. To see this, note that the run length of $A_{r_{i+1}}$ is $r_{i+1}  =\frac{l_{i+1}}{\eps^2}$ and the length of the projection of $v_i$ in $A_{r_i}$ under the matching that yields the optimal $\adv(v_i, A_{r_i})$ is no more than $2|v_i|=2l_i$ (See~\cref{lem:advProperties}). 
Therefore, $|U_e|\leq \frac{2 l_i}{r_{i+1}}$ and consequently no more that a $\frac{2 l_i / r_{i+1}}{l_i/l_{i+1}} = 2\eps^2$ fraction of strings in $U$ might be mapped to a substring of $A_{r_{i+1}}$ that crosses the border of some $0^{r_{i+1}}$ and $1^{r_{i+1}}$ intervals. 

Secondly, for any $j\in\{0, 1\}$, $\Pr\{T=j|v_i\} \geq \frac{\adv(v_i, A_{r_{i+1}})-8\eps^2}{8}$. This can be showed as follows. Let $M_i^j$ represent the subset of pairs of $M_i$ with one end in $U_j$ for $j\in\{0, 1, e\}$ and $v_i\rightarrow M_i$ represent the substring of $A_{r_{i+1}}$ where $v_i$ is projected under $M_i$. 
Note that $\Pr\{T=j|v_i\} = \frac{|U_j|}{|U|} = \frac{|U_j|\cdot l_i}{|v_i|}\geq \frac{|M_i^j|}{|v_i|} \geq \frac{|M_i^j|}{2|v_i\rightarrow M_i|}$.
Assume for contradiction that $\Pr\{T=j|v_i\} < \frac{\adv(v_i, A_{r_{i+1}})-8\eps^2}{8}$ for some $j$. Then, $|M_i^j| < |v_i\rightarrow M_i|\frac{\adv(v_i, A_{r_{i+1}})-8\eps^2}{4}$, which since $|M_i^{j'}|\leq\frac{|v_i\rightarrow M_i|}{2}$ for $j'\in\{0, 1\}$ and $|M^e_i|\leq 2\eps^2|v_i\rightarrow M_i|$, gives that $|M_i|<|v_i\rightarrow M_i|\left(\frac{1}{2} + 2\eps^2 + \frac{\adv(v_i, A_{r_{i+1}})-8\eps^2}{4}\right) =
|v_i\rightarrow M_i|\left(\frac{1}{2} + \frac{\adv(v_i, A_{r_{i+1}})}{4}\right)$.
However,
\noSTOC{
\begin{equation*}
\adv_{M_i} = \frac{3|M_i|-|v_i|-|p|}{|v_i|}
\Rightarrow 2|M_i|-|p|\geq |v_i|\adv_{M_i}\Rightarrow |M_i| \geq |p|\left(\frac{1}{2}+\frac{\adv_{M_i}}{4}\right).
\end{equation*}}
\STOConly{$\adv_{M_i} = \frac{3|M_i|-|v_i|-|p|}{|v_i|}
\Rightarrow 2|M_i|-|p|\geq |v_i|\adv_{M_i}\Rightarrow |M_i| \geq |p|\left(\frac{1}{2}+\frac{\adv_{M_i}}{4}\right).$}
This contradiction implies that $\Pr\{T=j|v_i\} \geq \frac{\adv(v_i, A_{r_{i+1}})-8\eps^2}{8}$.

The third and final important ingredient is provided by the following lemma that we prove later on.
\begin{lemma}\label{lem:polarization}
The following holds true:
$$\Big|\E\left[B_{i+1}|v_i, T=0\right] - \E\left[B_{i+1}|v_i, T=1\right]\Big| \geq  \frac{\adv(v_i, A_{r_{i+1}}) - 5\eps^2}{3}$$
\end{lemma}

To summarize, the above three properties imply that we have a three-valued random variable where the probability for one value is minuscule and there is at least $[\adv(v_i, A_{r_{i+1}}) - 5\eps^2]/3$ difference between the other two values each occurring with adequately large probabilities. This is enough for us to bound below the variance of such random variable. The following straightforward lemma abstracts this.
\begin{lemma}\label{lem:threeValuedVariance}
Let $X$ be a random variable that can take values $a_0$, $a_1$, and $a_2$ where
$\Pr\{X=a_i\} \geq \xi$ for $i\in\{0, 1\}$. Then, we have that 
$\Var(X) \geq \frac{\xi}{2}(a_0-a_1)^2$.
\end{lemma}
\noSTOC{\begin{proof}
$\Var(X)=\sum_{a_i}\Pr\{X=a_i\}(a_i-\bar{X})^2\geq \xi \left[(a_0-\bar{X})^2+(a_1-\bar{X})^2\right] \geq \frac{\xi}{2}(a_0-a_1)^2$.
\end{proof}}
Applying \cref{lem:threeValuedVariance} to our random variable gives that:
\noSTOC{\begin{equation*}
\Var_{T}(\E\left[B_{i+1}|v_i, T\right]) \geq \frac{1}{144}\left(\adv(v_i, A_{r_{i+1}})-8\eps^2\right)\left(\adv(v_i, A_{r_{i+1}}) - 5\eps^2\right)^2
\end{equation*}}
\STOConly{
$$\Var_{T}(\E\left[B_{i+1}|v_i, T\right]) \geq \frac{\left(\adv(v_i, A_{r_{i+1}})-8\eps^2\right)\left(\adv(v_i, A_{r_{i+1}}) - 5\eps^2\right)^2}{144}
$$}
Note the right hand side of this inequality is negative when $\adv(v_i, A_{r_{i+1}})\leq8\eps^2$. Therefore, we define function $g(x)$ as a function that takes value of $\frac{(x-8\eps^2)(x-5\eps^2)}{144}$ when $x > 8\eps^2$ and zero otherwise. Note that $g$ is a convex function. We have that 
\begin{equation}
\Var_{T}(\E\left[B_{i+1}|v_i, T\right]) \geq g(\adv(v_i, A_{r_{i+1}}))\label{lem:3ValVarLowerBound}
\end{equation}

Plugging \eqref{lem:3ValVarLowerBound} into \eqref{eqn:varLoweboundStep2} gives that 
\noSTOC{\begin{eqnarray}
\E\left[\Var(B_{i+1}|v_{i})\right] &\geq& \E_{v_i}\big[\Var_{T}(\E\left[B_{i+1}|v_i, T\right])\big]
\geq\E_{v_i}\big[g(\adv(v_i, A_{r_{i+1}}))\big]\nonumber\\
&\geq&g\left(\E_{v_i}\big[\adv(v_i, A_{r_{i+1}})\big]\right) = g(\E[A_{i+1}])\label{eqn:jensen}\\
&\geq& g\left(\frac{\eps}{2}\right) = \frac{\eps^3}{1152}+o(\eps^3)\label{eqn:lastStepVarLower}
\end{eqnarray}}
\STOConly{\begin{eqnarray}
\E\left[\Var(B_{i+1}|v_{i})\right] &\geq& \E_{v_i}\big[\Var_{T}(\E\left[B_{i+1}|v_i, T\right])\big] \nonumber\\&
\geq &\E_{v_i}\big[g(\adv(v_i, A_{r_{i+1}}))\big]\nonumber\\
&\geq&g\left(\E_{v_i}\big[\adv(v_i, A_{r_{i+1}})\big]\right) = g(\E[A_{i+1}])\label{eqn:jensen}\\
&\geq& g\left(\eps/2\right) = \eps^3/1152+o(\eps^3)\label{eqn:lastStepVarLower}
\end{eqnarray}}
where \eqref{eqn:jensen} follows from the Jensen inequality and \eqref{eqn:lastStepVarLower} follows from \cref{lem:process-guarantees} and the fact that $g$ is an increasing function. Note that the right hand side is at least $\frac{\eps}{1200}$ for sufficiently small $\eps$. This completes the proof of \cref{lem:process-variance-guarantee} (With the exception of \cref{lem:polarization}).
\end{proof}

With \cref{lem:process-variance-guarantee} proved, one can easily prove \cref{thm:LowRateListDecBinaryCodes}.

\begin{proof}[{\bf Proof of \cref{thm:LowRateListDecBinaryCodes}}]
Since $\Var(B_{i+1}) \geq \Var(B_i) + \eps^3/1200$, we have that
\noSTOC{$$\Var(B_{k})\geq\Var(B_1) + (k-1)\frac{\eps^3}{1200} \geq \frac{(k-1)\eps^3}{1200}.$$}
\STOConly{$\Var(B_{k})\geq\Var(B_1) + (k-1)\frac{\eps^3}{1200} \geq \frac{(k-1)\eps^3}{1200}.$}
If $k > \frac{1200}{\eps^3}$, the above inequality implies that $\Var(B_{k}) > 1$ which is impossible since $B_{k}$ takes value in $[-1, 1]$. This contradiction implies that the list size $k \leq \frac{1200}{\eps^3}$.
\end{proof}

We now proceed to the proof of \cref{lem:polarization}.

\subsection{Proof of \cref{lem:polarization}}
\pushQED{\qed} 
Consider $v_i$ and the matching that yields the optimal advantage from $v_i$ to $A_{r_{i+1}}$, denoted by $M_i$. We denote the substring of $A_{r_{i+1}}$ that is identified by the projection of $v_i$ under $M_i$ as $p = v_i\rightarrow M_i$.
To simplify the analysis, we perform a series of transformations on $v_i$, $M_i$, and $p$ that does not decrease $\adv_{M_i}$ except by a small quantity. \cref{fig:transformtion} depicts the steps of this transformation described below.

\begin{figure}[]
  \centering
  \includegraphics[height=4in]{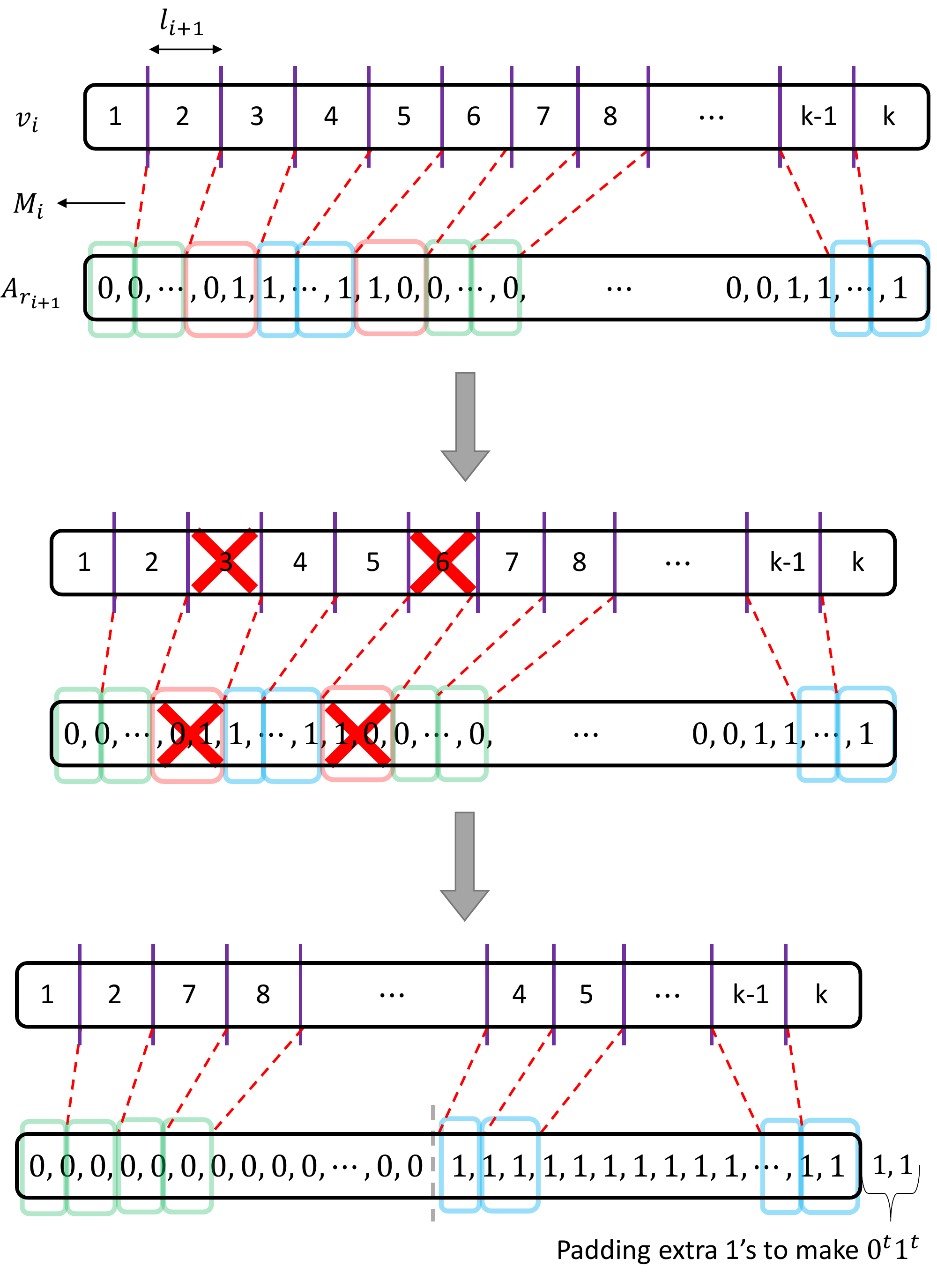}\\
  \caption{Three steps of transformation in \cref{lem:polarization}.}\label{fig:transformtion}
\end{figure}

\begin{enumerate}
\item \label{step:modificationStepOne}First, we delete all substrings of $U_e$---i.e., substrings of length $l_i$ in $v_i$ whose projection contain both \noSTOC{zeros and ones}\STOConly{0s and 1s}---from $v_i$.
\item We reorder the substrings of length $l_{i+1}$ in $v_i$ by shifting all $U_0$ substrings together and all $U_1$ substrings together. We accordingly shift the projections of these strings in $p$ to the similar order. This was, the remainder of $M_i$ from step \ref{step:modificationStepOne} will be preserved as a valid matching between reordered strings.
\item At this point, string $p$ consists of a stretch of zeros followed by a stretch of ones. If the length of two stretches are not equal, we add adequate zeros or ones to the smaller stretch to make $p$ have the form of $0^t1^t$.
\end{enumerate}

To track the changes in $\adv_{M_i}$ during this transformation, we track how $|M_i|$, $|v_i|$ and $|p|$ change throughout the three steps mentioned above.

In the first step, a total of up to 
$|U_e| l_{i+1}$
elements are removed from $v_i$ and $M_i$. Note that since the run length of $A_{r_{i+1}}$ is $r_{i+1}$, there can only be $\frac{|p|}{r_{i+1}}$ substrings in $U_e$. Therefore,
\noSTOC{$$|U_e| l_{i+1} \leq \frac{|p|l_{i+1}}{r_{i+1}} = |p|\eps^2 \leq 2\eps^2|v_i|.$$}\STOConly{$|U_e| l_{i+1} \leq \frac{|p|l_{i+1}}{r_{i+1}} = |p|\eps^2 \leq 2\eps^2|v_i|.$}

The second step preserves $|M_i|$, $|v_i|$ and $|p|$.

Finally, since $p$ is a substring of $A_{r_{i+1}}$, the third step increases $|p|$ only by up to $r_{i+1}$. Note the run length of the $A_{r_{i+1}}$s and consequently $l_{i+1}$s are different by a multiplicative factor of at least $\frac{1}{\eps^4}$ by the definition of the code $\mathcal{C}$. Therefore,
$r_{i+1} = \frac{l_{i+1}}{\eps^2} = \frac{l_{i+1}|v_i|}{\eps^2|v_i|} = \frac{l_{i+1}|v_i|}{\eps^2 l_{i}}
\leq \eps^2|v_i|$. 

Overall, the value of the $\adv_{M_i} = \frac{3|M|-|p|-|v_i|}{|v_i|}$ can be affected by a maximum of $(3-1)\times2\eps^2|v_i| + \eps^2|v_i| = 5\eps^2|v_i|$ decrease in the numerator and $\eps^2|v_i|$ decrease in the denominator. Therefore, the eventual advantage does not drop below $\adv_{M_i} - 5\eps^2$.
Let us denote the transformed versions of $v_i$, $p$, and $M_i$ by $\bar{v}_i$, $\bar{p}$, and $\bar{M}_i$ respectively. We have shown that
\begin{equation}
\adv_{\bar{M}_i} \geq \adv_{M_i} - 5\eps^2.\label{eqn:advLowerBoundAfterTransformation}
\end{equation}
Further, let $\bar{v}_i = (\bar{v}_i^0, \bar{v}_i^1)$ so that $\bar{v}_i^0$ and $\bar{v}_i^1$ respectively correspond to the part of $\bar{v}_i$ that is mapped to $0^t$ and $1^t$ under $\bar{M}_i$.
Consider the matching between $\bar{v}_i$ and $\bar{p}$ that connects as many zeros as possible between the $\bar{v}_i^0$ and $0^t$ and as many ones as possible between the $\bar{v}_i^1$ to $1^t$ portion of $\bar{p}$. Clearly, the size of $\bar{M}_i$ cannot exceed the size of this matching and therefore, 
\begin{equation}\label{eqn:01advantage}
\adv_{\bar{M}_i} \leq \frac{3\left[\min\{t, \Count_0(\bar{v}_i^0)\}+\min\{t, \Count_1	(\bar{v}_i^1)\}\right] - |\bar{v}_i| - 2t}{|\bar{v}_i|}
\end{equation}
Note that as long as $t<\Count_0(\bar{v}_i^0)$ or $t<\Count_1(\bar{v}_i^1)$, increasing $t$ in the right hand side term does not make it smaller. Therefore, the inequality \eqref{eqn:01advantage} holds for $t=\max_{j\in\{0, 1\}}\{\Count_j(\bar{v}_i^j)\}$.
Without loss of generality, assume that $\Count_0(\bar{v}_i^0) \leq \Count_1(\bar{v}_i^1)$ and set $t=\Count_1(\bar{v}_i^1)$. Then we have the following.
\noSTOC{\begin{eqnarray}
&&\adv_{\bar{M}_i} \leq \frac{3 \Count_0(\bar{v}_i^0)+\Count_1	(\bar{v}_i^1) - |\bar{v}_i|}{|\bar{v}_i|}\nonumber\\
&\Rightarrow& \adv_{\bar{M}_i} \leq \frac{3 \frac{1-\bias(\bar{v}_i^0)}{2}|\bar{v}_i^0|+\frac{1+\bias(\bar{v}_i^1)}{2}|\bar{v}_i^1| - (|\bar{v}_i^0|+|\bar{v}_i^1|)}{|\bar{v}_i|}\nonumber\\
&\Rightarrow& 2\adv_{\bar{M}_i}|\bar{v}_i| \leq 3 (1-\bias(\bar{v}_i^0))|\bar{v}_i^0|+(1+\bias(\bar{v}_i^1))|\bar{v}_i^1| - 2(|\bar{v}_i^0|+|\bar{v}_i^1|)\nonumber\\
&\Rightarrow& 2\adv_{\bar{M}_i}|\bar{v}_i| \leq 
\left[1-3\bias(\bar{v}_i^0)\right] |\bar{v}_i^0| - \left[1-\bias(\bar{v}_i^1)\right] |\bar{v}_i^1|\label{eqn:biasDiscrepency}
\end{eqnarray}}
\STOConly{
\begin{align}
&\adv_{\bar{M}_i} \leq \frac{3 \Count_0(\bar{v}_i^0)+\Count_1	(\bar{v}_i^1) - |\bar{v}_i|}{|\bar{v}_i|}\nonumber\\
\Rightarrow\ & \adv_{\bar{M}_i} \leq \frac{3 \frac{1-\bias(\bar{v}_i^0)}{2}|\bar{v}_i^0|+\frac{1+\bias(\bar{v}_i^1)}{2}|\bar{v}_i^1| - (|\bar{v}_i^0|+|\bar{v}_i^1|)}{|\bar{v}_i|}\nonumber\\
\Rightarrow\ & 2\adv_{\bar{M}_i}|\bar{v}_i| \leq 3 (1-\bias(\bar{v}_i^0))|\bar{v}_i^0|\nonumber\\
&+(1+\bias(\bar{v}_i^1))|\bar{v}_i^1| - 2(|\bar{v}_i^0|+|\bar{v}_i^1|)\nonumber\\
\Rightarrow\	 & 2\adv_{\bar{M}_i}|\bar{v}_i| \leq 
\left[1-3\bias(\bar{v}_i^0)\right] |\bar{v}_i^0| - \left[1-\bias(\bar{v}_i^1)\right] |\bar{v}_i^1|\label{eqn:biasDiscrepency}
\end{align}
}
We claim that the above inequality leads to the fact that $|\bias(\bar{v}_i^1)- \bias(\bar{v}_i^0)| \geq \adv_{\bar{M}_i}/3$. Assume for contradiction that this is not the case. Therefore, replacing the term $\bias(\bar{v}_i^0)$ with $\bias(\bar{v}_i^1)$ in \eqref{eqn:biasDiscrepency} does not change the value of the right hand side by any more than $|\bar{v}_i|\cdot \adv_{\bar{M}_i}$. Same holds true with replacing the term $\bias(\bar{v}_i^1)$ with $\bias(\bar{v}_i^0)$ in \eqref{eqn:biasDiscrepency}. 	This implies that, with $b^* =\max\{\bias(\bar{v}_i^0), \bias(\bar{v}_i^1)\}$, we have that
\noSTOC{\begin{eqnarray}
&&\adv_{\bar{M}_i}|\bar{v}_i| \leq 
\left(1-3b^*\right)\cdot |\bar{v}_i^0| - \left(1-b^*\right) |\bar{v}_i^1|\nonumber\\
&\Rightarrow& \left(1-b^*\right) |\bar{v}_i^1|<\left(1-3b^*\right) |\bar{v}_i^0|\label{eqn:advantaegAssumption}
\end{eqnarray}}
\STOConly{$\adv_{\bar{M}_i}|\bar{v}_i| \leq 
\left(1-3b^*\right)\cdot |\bar{v}_i^0| - \left(1-b^*\right) |\bar{v}_i^1|$ and, therefore,
\begin{align}
\left(1-b^*\right) |\bar{v}_i^1|<\left(1-3b^*\right) |\bar{v}_i^0|\label{eqn:advantaegAssumption}
\end{align}}
On the other hand, we assumed earlier (without loss of generality) that $\Count_0(\bar{v}_i^0) \leq \Count_1(\bar{v}_i^1)$. Therefore,
\STOConly{\begin{eqnarray}
&&\left(1-\bias(\bar{v}_i^0)\right) |\bar{v}_i^0|\leq\left(1+\bias(\bar{v}_i^1)\right) |\bar{v}_i^1|\nonumber\\
&\Rightarrow& 
\left(1-b^*\right) |\bar{v}_i^0|\leq\left(1+b^*\right) |\bar{v}_i^1|\label{eqn:WLOGassumption}
\end{eqnarray}}
\noSTOC{\begin{eqnarray}
&&\Count_0(\bar{v}_i^0) \leq \Count_1(\bar{v}_i^1)\nonumber\\
&\Rightarrow& 
\left(1-\bias(\bar{v}_i^0)\right) |\bar{v}_i^0|\leq\left(1+\bias(\bar{v}_i^1)\right) |\bar{v}_i^1|\nonumber\\
&\Rightarrow& 
\left(1-b^*\right) |\bar{v}_i^0|\leq\left(1+b^*\right) |\bar{v}_i^1|\label{eqn:WLOGassumption}
\end{eqnarray}}
Note that since $|b^*|\leq 1$, $(1-b^*)^2 > (1+b^*)(1-3b^*) \Rightarrow \frac{1-3b^*}{1-b^*} < \frac{1-b^*}{1+b^*}$. Multiplying the two sides of this inequality to the sides of \eqref{eqn:WLOGassumption} gives that 
\noSTOC{$$\left(1-3b^*\right) |\bar{v}_i^0|\leq\left(1+b^*\right) |\bar{v}_i^1|$$}
\STOConly{$\left(1-3b^*\right) |\bar{v}_i^0|\leq\left(1+b^*\right) |\bar{v}_i^1|$}
which contradicts \eqref{eqn:advantaegAssumption}. Therefore, we must have 
\noSTOC{$$|\bias(\bar{v}_i^1)- \bias(\bar{v}_i^0)| \geq \adv_{\bar{M}_i}/3.$$}
\STOConly{$|\bias(\bar{v}_i^1)- \bias(\bar{v}_i^0)| \geq \adv_{\bar{M}_i}/3.$}
Note that $\bias(\bar{v}_i^j) = \E\left[B_{i+1}|v_i, T=j\right]$ since $\bias(\bar{v}_i^j)$ is the average bias of all strings in $U_j$. Therefore, combining with \eqref{eqn:advLowerBoundAfterTransformation}, we have that
\begin{equation*}
\Big|\E\left[B_{i+1}|v_i, T=0\right] - \E\left[B_{i+1}|v_i, T=1\right]\Big| \geq \frac{\adv(v_i, A_{r_{i+1}}) - 5\eps^2}{3}.\qedhere
\end{equation*}
\popQED

\section{Proof of \noSTOC{\cref{thm:main}}\STOConly{Theorem \ref{thm:main}}: Concatenated InsDel Codes}\label{sec:concatenation}
We recall that the concatenation of an inner insdel code $\mathcal{C}_{\text{in}}$ over an alphabet of size $|\Sigma_{\text{in}}|$ and an outer insdel code, $\mathcal{C}_{\text{out}}$, over an alphabet of size $|\Sigma_{\text{out}}| = |\mathcal{C}_{\text{in}}|$ as a code over alphabet $\Sigma_{\text{in}}$, is obtained by taking each codeword $x \in \mC_\textout$, encoding each symbol of $x$ with $\mathcal{C}_{\text{in}}$, and appending the encoded strings together to obtain each codeword of the concatenated code.

In this section, we will show that, concatenating an inner code $\mC_\textin$ from \cref{thm:LowRateListDecBinaryCodes} that can $L_\textin$-list decode from any $\gamma$ fraction of insertions and $\delta$ fraction deletions when 
$2\delta + \gamma < 1-\eps_\textin$ along with an appropriately chosen outer code $\mC_\textout$ from \cref{thm:LargeAlphaInsDelListDecoding}, one can obtain an infinite family of constant-rate insertion-deletion codes that are efficiently list-decodable from any $\gamma$ fraction of insertions and $\delta$ fraction of deletions as long as 
$2\delta + \gamma < 1-\eps$ for $\eps=\frac{16}{5}\eps_\textin$.

\subsection{Construction of the Concatenated Code}
We start by fixing some notation. Let $\mC_\textout$ be able to $L_\textout$-list decode from $\delta_{\text{out}}$ fraction of deletions and $\gamma_{\text{out}}$ fraction of insertions. Further, let us indicate the block sizes of $\mC_\textout$ and $\mC_\textin$ with $n_\textout$ and $n_\textin=\left\lceil\log \left|\Sigma_{\text{out}}\right|\right\rceil$.

To construct our concatenated codes, we utilize \cref{thm:LargeAlphaInsDelListDecoding} to obtain an efficient family of codes $\mathcal{C}_{\text{out}}$ over alphabet $\Sigma_{\text{out}}$ of size $O_{\gamma_{\text{out}}, \delta_{\text{out}}}(1)$ that is $L_{\text{out}}$-list decodable from any $\delta_{\text{out}}$ fraction of deletions and $\gamma_{\text{out}}$ fraction of insertions for appropriate parameters $\delta_{\text{out}}$ and $\gamma_{\text{out}}$ that we determine later. We then concatenate any code in $\mathcal{C}_{\text{out}}$ with an instance of the binary list-decodable codes from \cref{thm:LowRateListDecBinaryCodes}, $C_{\text{in}}$, with parameter $n_{\text{in}}=\left\lceil\log \left|\Sigma_{\text{out}}\right|\right\rceil$ and a properly chosen $\eps_{\text{in}}$. 
We will determine appropriate values for all these parameters given $\eps$ when describing the decoding procedure in \cref{sec:decoding}. \cref{fig:parameters} shows the order of determining all parameters. We remark that the following two properties for the utilized inner and outer codes are critical to this order of fixing parameters:
\begin{enumerate}
    \item The alphabet size of the family of codes used as the outer code only depends on $\delta_\textout$ and $\gamma_\textout$ and is independent of the outer block size $n_\textout$. (See \cref{thm:LargeAlphaInsDelListDecoding})
    \item The list size of the family of codes used as the inner code, $L_\textin$, merely depends on parameter $\eps_{\text{in}}$ in \cref{thm:LowRateListDecBinaryCodes} and is independent of the size of the code or its block length, i.e., $|\mC_\textin|$ or $n_\textin$.
\end{enumerate}
\begin{figure}[]
  \centering
  \noSTOC{\includegraphics[width=.7\textwidth]{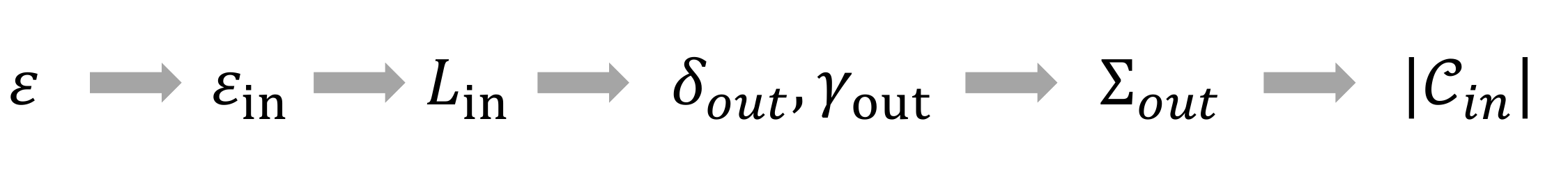}\\}
  \STOConly{\includegraphics[width=\linewidth]{Parameters.pdf}\\}
  \caption{The order of determining parameters in the proof of \cref{thm:main}.}\label{fig:parameters}
\end{figure}

\subsection{Decoding Procedure and \noSTOC{Determining }Parameters}\label{sec:decoding}

    We now analyze the resulting family of codes and choose the undetermined parameters along the way of describing the decoding procedure. A pseudo-code of the decoding procedure is available in \cref{alg:conctenatedDecoder}. Let $\mC$ be a binary code with block length $n$ that is obtained from the above-mentioned concatenation. Take the codeword $x\in C$ and split it into \emph{blocks} of length $n_{\text{in}}$. Note that each such block corresponds to the encoding of some symbol in $\Sigma_{\text{out}}$ under $\mC_{\text{in}}$. Let $x'$ be a string obtained by applying $n\gamma$ insertions and $n\delta$ deletions into $x$ where $n=n_\textin n_\textout$ and $\gamma+2\delta < 1-\eps$. For each block of $x$, we define the error count to be the total number of insertions that have occurred in that block plus twice the number of deleted symbols in it. Clearly, the average value of error count among all blocks is $n_\textin(\gamma+2\delta) < n_\textin(1-\eps)$. By a simple averaging, at least $\left(1-\frac{1-\eps}{1-\eps/4}\right) n_{\text{out}} \ge \frac{3\eps}{4}\cdot n_{\text{out}}$ of those blocks have an error count of $n_\textin(1-\frac{\eps}{4})$ or less. 
Let us call the set of all such blocks $S$. 

Further, we partition $S$ into smaller sets based on the number of deletions occurring in the blocks of $S$. Let $S_i\in S$ be the subset of blocks in $S$ for which the number of deletions is in 
\noSTOC{$\left[n_{\text{in}}\cdot\frac{\eps}{16}\cdot (i-1), n_{\text{in}}\cdot \frac{\eps}{16}\cdot i\right)$}
\STOConly{$[n_{\text{in}}\cdot\frac{\eps}{16}\cdot (i-1), n_{\text{in}}\cdot \frac{\eps}{16}\cdot i)$}
 for $i=1, 2, \cdots, 8/\eps$\footnote{Note that the fraction of deletions cannot exceed $\frac{1}{2}$ assuming $n_\textin(\gamma+2\delta) < n_\textin(1-\eps)$.}. The following \noSTOC{two properties }hold true:
\begin{enumerate}
    \item All blocks in $S_i$ suffer from at least $n_{\text{in}}\cdot\frac{\eps}{16}\cdot (i-1)$ deletions. Further, they can suffer from up to $n_{\text{in}}\cdot\left(1-\frac{\eps}{4}-\frac{2\eps}{16}\cdot (i-1)\right)$ insertions. Therefore, they all appear as substrings of length $n_{\text{in}}\cdot\left(2-\frac{\eps}{4}-\frac{3\eps}{16}\cdot (i-1)\right)$ or less in $x'$.
    \item We have that $S=\dot\bigcup_{i=1}^{8/\eps} S_i$. By the Pigeonhole principle, for some $i^*\in\left[1, 8/\eps\right]$, $|S_{i^*}|\ge \frac{3\eps^2}{32}n_{\text{out}}$.
\end{enumerate}

Our decoding algorithm consists of $8/\eps$ rounds each consisting of two phases of inner and outer decoding. During the first phase of each round $i=1,2,\cdots, 8/\eps$, the algorithm uses the decoder of the inner code on $x'$ to construct a string $T_i$ over alphabet $\Sigma_\textout$ and then, in the second phase, uses the decoder of the outer code on input $T_i$ to obtain a list $List_i$ of size $L_\textout$. In the end, the decoding algorithm outputs the union of all such lists $\bigcup_i List_i$.

\begin{algorithm}
\caption{Decoder of the Concatenated Code}\label{alg:conctenatedDecoder}
\begin{algorithmic}[1]
\noSTOC{\Procedure{Concatenated-Decoder}{$x', \eps, n_\textin, n_\textout, \Decode_{\mC_\textin}(\cdot), \Decode_{\mC_\textout}(\cdot)$}}
\STOConly{\Procedure{Concat'd-Dec}{$x', \eps, n_\textin, n_\textout, \Decode_{\mC_\textin}, \Decode_{\mC_\textout}$}}

\medskip 

\State {Output $\leftarrow\emptyset$}

\For{$i\in\left\{1, 2, \cdots, \frac{8}{\eps}\right\}$}\Comment{Round $i$}
\State $w\leftarrow\left\lfloor\frac{n_{\text{in}}(2-\eps/4-3\eps(i-1)/16)}{n_{\text{in}} \eps/16}\right\rfloor+1$\noSTOC{\Comment{Length of the sliding window is $w\cdot\frac{n_\textin\eps}{16}$.}}

\State $T_i \leftarrow $ empty string

\medskip 

\For{$j\in \left\{1, 2, \cdots, \frac{|x'|}{n_{\textin}\eps/16}-w\right\}$}\Comment{Phase I\noSTOC{: Inner Decoding}}
\State {$List \leftarrow \Decode_{\mC_\textin}\left(x'\left[\frac{n_{\text{in}} \eps}{16}\cdot j, \frac{n_{\text{in}} \eps}{16}\cdot (j + w) \right]\right)$}
\State {Pad symbols of $\Sigma_\textout$ corresponding to the elements of $List$ to the right of $T_i$.}
\EndFor

\medskip 

\State {Output $\leftarrow$ Output $\cup$  $\Decode_{\mC_\textout}\left(T_i\right)$}\Comment{Phase II\noSTOC{: Outer Decoding}}
\EndFor

\State {\bf return} Output

\EndProcedure
\end{algorithmic}
\end{algorithm}

\paragraph{Description of Phase I (Inner Decoding)} 
We now proceed to the description of the first phase in each round $i\in\{1,2,\cdots, 8/\eps\}$. In the construction of $T_i$, we aim for correctly decoding the blocks in $S_i$. As mentioned above, all such blocks appear in $x'$ in a substring of length $n_{\text{in}}\cdot\left(2-\frac{\eps}{4}-\frac{3\eps}{16}\cdot (i-1)\right)$ or less.

Having this observation, we run the deocoder of the inner code on substrings of $x'$ of form 
$x'\left[\frac{n_{\text{in}} \eps}{16}\cdot j, \frac{n_{\text{in}} \eps}{16}\cdot (j + w) \right]$ for all $j=1,2,\cdots, \frac{|x'|}{n_{\textin}\eps/16}-w$ where 
\noSTOC{$$w=\left\lfloor\frac{n_{\text{in}}(2-\eps/4-3\eps(i-1)/16)}{n_{\text{in}} \eps/16}\right\rfloor + 1.$$}
\STOConly{$w=\left\lfloor\frac{n_{\text{in}}(2-\eps/4-3\eps(i-1)/16)}{n_{\text{in}} \eps/16}\right\rfloor + 1.$}
One can think of such substrings as a \emph{window} of size $w\cdot\frac{n_\textin \eps}{16}$ that slides in $\frac{n_\textin \eps}{16}$ increments.

Note that each block $B$ in $S_i$ appears within such window and is far from it by, say, $D_B$ deletions and no more than $n_{\text{in}}\left(1-\frac{\eps}{4}\right) - 2D_B  + \frac{n_{\text{in}} \eps}{16}$ insertions where the additional $\frac{n_{\text{in}} \eps}{16}$ term in insertion count comes from the extra symbols around the block in the fixed sized window. As long as the fraction of insertions plus twice the fraction of deletions that are needed to convert a block of $S_i$ into its corresponding window does not exceed $1-\eps_{\text{in}}$, the output of the inner code's decoder for input 
$x'\left[\frac{n_{\text{in}} \eps}{16}\cdot j, \frac{n_{\text{in}} \eps}{16}\cdot (j + w) \right]$
will contain the block $B$ of $S_i$. So, we choose $\eps_{\text{in}}$ such that 
\begin{eqnarray}
&&n_{\text{in}}\left(1-\frac{\eps}{4}\right) - 2D_B  + \frac{n_{\text{in}} \eps}{16} + 2D_B \le n_{\text{in}}(1-\eps_{\text{in}})\label{eqn:eps-in-choice}\\
&\Leftrightarrow&n_{\text{in}}(1-3\eps/16) \le n_{\text{in}}(1-\eps_{\text{in}})\noSTOC{\nonumber\\
&}\Leftrightarrow\noSTOC{&}\eps_{\text{in}}\le \frac{3}{16}\eps\nonumber
\end{eqnarray}

Now, each element in the output list corresponds to some codeword of the inner code and, therefore, some symbol in $\Sigma_{\text{out}}$. For each run of the decoder of the inner code, we take the corresponding symbols of $\Sigma_{\text{out}}$ and write them back-to-back in arbitrary order. Then, we append all such strings in the increasing order of $j$ to obtain $T_i$.

\paragraph{Description of Phase II (Outer Decoding)} 
Note that the length of $T_i$ is at most 
$\frac{|x'|}{n_{\text{in}}\eps/16} L_{\text{in}} \leq
\frac{2n_{\text{in}}n_{\text{out}}}{n_{\text{in}}\eps/16} L_{\text{in}} = n_{\text{out}}\cdot\frac{32}{\eps}L_{\text{in}}$. Further, $T_i$ contains symbols corresponding to all blocks of $S_i$ as a subsequence (i.e., in the order of appearance) except possibly the ones that appear in the same run of the inner decoder together. Since the fraction of deletions happening to each block in $S_i$ is less than $\frac{1}{2}$ and the size of the inner decoding sliding window is no more than $2n_\textin$, the number of blocks of $S_i$ that can appear in the same window in the first phase is at most 4. This gives that $T_{i}$ has a common subsequence of size at least $\frac{|S_i|}{4}$ with the codeword of the outer code.

We mentioned earlier that for some $i^*$, $|S_{i^*}|\ge \frac{3\eps^2}{32}n_{\text{out}}$. Therefore, for such $i^*$, $T_{i^*}$ is different from $x$ by up to a $1-\frac{3\eps^2}{128}$ fraction of deletions and $\frac{32}{\eps}L_{\text{in}}$ fraction of insertions. Therefore, by taking $\delta_{\text{out}}=1-\frac{3\eps^2}{128}$,  $\gamma_{\text{out}}=\frac{32}{\eps}L_{\text{in}}=O\left(\frac{1}{\eps^4}\right)$, and using each $T_i$ as an input to the decoder of the outer code in the second phase, $x$ will certainly appear in the outer output list for some $T_i$. (Specifically, for $i=i^*$.)

\subsection{Remaining Parameters}
As shown in \cref{sec:decoding}, we need a list-decodable code as outer code that can list-decode from $\delta_{\text{out}}=1-\frac{3\eps^2}{128}$ fraction of deletions and  $\gamma_{\text{out}}=\frac{32}{\eps}L_{\text{in}}=O\left(\frac{1}{\eps^4}\right)$ fraction of insertions. To obtain such codes we use \cref{thm:LargeAlphaInsDelListDecoding} with parameters $\gamma=\frac{32}{\eps}L_{\text{in}}$ and $\epsilon=\frac{3\eps^2}{256}$.
This implies that the rate of the outer code is $r_\textout = \frac{3\eps^2}{256} = O(\eps^2)$, it is $L_\textout=O_\eps(\exp(\exp(\exp(\log^*n))))$ list-decodable, and can be defined over an alphabet size of $|\Sigma_\textout|=e^{O\left(\frac{1}{\eps^{10}}\log \frac{1}{\eps^8}\right)}$.

Consequently, $|\mC_{in}|=\log|\Sigma_{\text{out}}|=O\left(\frac{1}{\eps^{10}}\log \frac{1}{\eps}\right)$. Note that in \cref{thm:LowRateListDecBinaryCodes}, the block length of the inner code can be chosen independently of its list size as the list size only depends on $\eps_\textin$. This is a crucial quality in our construction since in our analysis $\eps_\textin$ and $L_\textin$ are fixed first and then $|C_\textin|$ is chosen depending on the properties of the outer code.

As the decoder of the outer code is used $\frac{8}{\eps}$ times in the decoding of the concatenated code, the list size of the concatenated code will be $L=\frac{8}{\eps}\cdot L_{\text{out}} = O_\eps(\exp(\exp(\exp(\log^*n))))$.
The rate of the concatenated code is 
\noSTOC{$$r=r_\textout r_\textin = O\left(\eps^2\cdot\frac{\log\log |\mC_\textin|}{n_\textin}\right) = \noSTOC{O\left(\eps^2\cdot\frac{\log\log |\mC_\textin|}{(1/\eps^4)^{|\mC_\textin|}}\right) =} e^{-O\left(\frac{1}{\eps^{10}}\log^2 \frac{1}{\eps}\right)}.$$}
\STOConly{$r=r_\textout r_\textin = O\left(\eps^2\cdot\frac{\log\log |\mC_\textin|}{n_\textin}\right) = \noSTOC{O\left(\eps^2\cdot\frac{\log\log |\mC_\textin|}{(1/\eps^4)^{|\mC_\textin|}}\right) =} e^{-O\left(\frac{1}{\eps^{10}}\log^2 \frac{1}{\eps}\right)}.$}

Finally, since the outer code is efficient and the inner code is explicit and can be decoded by brute-force in $O_\eps(1)$ time, the encoding and decoding procedures run in polynomial time. This concludes the proof of \cref{thm:main}.

\section{Extension to Larger Alphabets}\label{sec:large-alphabet}
In this section we extend the results presented so far to $q$-ary alphabets where $q > 2$. 
\subsection{Feasibility Region: Upper Bound}
For an alphabet of size $q$, no positive-rate family of deletion codes can protect against $1-\frac{1}{q}$ fraction of errors since, with that many deletions, an adversary can simply delete all but the most frequent symbol of any codeword. Similarly, for insertion codes, it is not possible to achieve resilience against $q-1$ fraction of errors as adversary would be able to turn any codeword $x\in q^n$ to $(1, 2, \cdots, q)^n$.

The findings of the previous sections on binary alphabets might suggest that the feasibility region for list-decoding is the region mapped out by these two points, i.e., $\frac{\delta}{1-\frac{1}{q}} + \frac{\gamma}{q-1} < 1$. However, this conjecture turns out to be false. The following theorem provides a family of counterexamples.

\begin{theorem}\label{thm:counterexamples}
For any alphabet size $q$ and any $i=1, 2, \cdots, q$, no positive-rate $q$-ary infinite family of insertion-deletion codes can list-decode from $\delta = \frac{q-i}{q}$ fraction of deletions and $\gamma = \frac{i(i-1)}{q}$ fraction of insertions.
\end{theorem}
\begin{proof}
Take a codeword $x\in [q]^n$. With $\delta n = \frac{q-i}{q} \cdot n$ deletions, the adversary can delete the $q-i$ least frequent symbols to turn $x$ into $x'\in\Sigma_d^{n(1-\delta)}$ for some $\Sigma_d = \{\sigma_1, \cdots, \sigma_{i}\} \subseteq [q]$. Then, with $\gamma n = n(1-\delta)(i - 1)=n\frac{i(i-1)}{q}$ insertions, it can turn $x'$ into $[\sigma_1, \sigma_2, \cdots, \sigma_{i}]^{n(1-\delta)}$. Such adversary only allows $O(1)$ amount of information to pass to the receiver. Hence, no such family of codes can yield a positive rate.
\end{proof}

Note that all points $(\gamma, \delta) = \left(\frac{i(i-1)}{q}, \frac{q-i}{q}\right)$ are located on a second degree curve inside the conjectured feasibility region $\frac{\delta}{1-\frac{1}{q}} + \frac{\gamma}{q-1} < 1$ (see \cref{fig:conjecture}). \noSTOC{Our next step is to show that the actual feasibility region is a subset of the polygon outlined by these points.}\STOConly{In the extended version of this paper, we use a simple time-sharing argument to show that the actual feasibility region is a subset of the polygon outlined by these points.}

\begin{figure}[]
  \centering
  \noSTOC{\includegraphics[width=.9\textwidth]{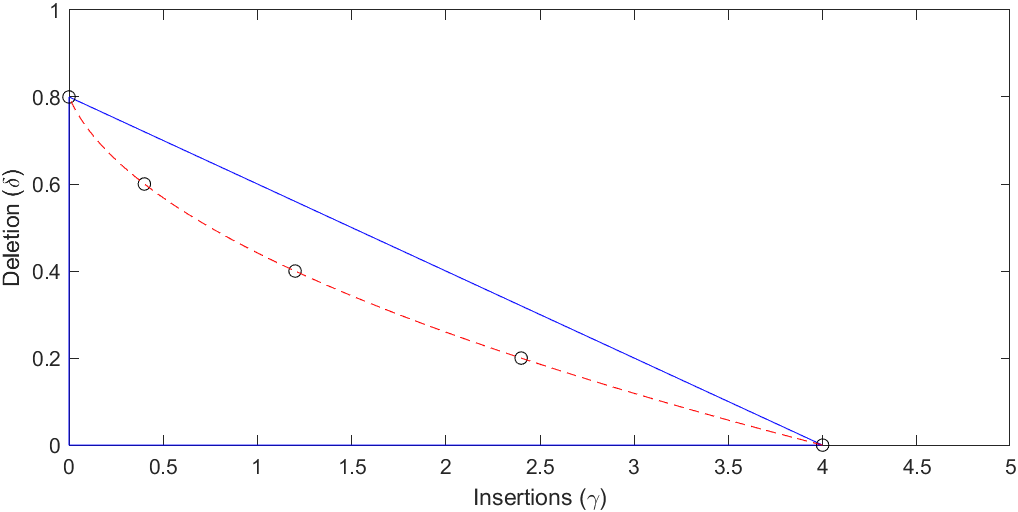}}
  \STOConly{\includegraphics[width=\linewidth]{Conjectured.png}}
  \caption{Infeasible points inside the conjectured feasibility region. (Illustrated for $q=5$)}\label{fig:conjecture}
\end{figure}

\begin{theorem}\label{thm:QaryUpperbound}
For any positive integer $q > 2$, define $F_q$ as the concave polygon defined over vertices
$\left(\frac{i(i-1)}{q}, \frac{q-i}{q}\right)$ for $i = 1, \cdots, q$ and $(0, 0)$. (see \cref{fig:actual-region}). $F_q$ does not include the border except the two segments 
$\left[(0, 0), (q-1, 0)\right)$ and $\left[(0, 0), \left(0, 1-\frac{1}{q}\right)\right)$. Then, for any pair of positive real numbers $(\gamma, \delta) \not\in F_q$, there exists no infinite family of $q$-ary codes with positive rate that can correct from $\delta$ fraction of deletions and $\gamma$ fraction of insertions.
\end{theorem}
\noSTOC{
\begin{proof}
In order to prove this, it suffices to show that for any pair of consecutive vertices on the polygon like 
$p_i = \left(\frac{i(i-1)}{q}, \frac{q-i}{q}\right)$ 
 and $p_{i+1} = \left(\frac{i(i+1)}{q}, \frac{q-i-1}{q}\right)$, 
 the entirety of the segment between $p_i$ and $p_{i+1}$ lie outside of the feasibility region. To this end, we show that for any $i=1, 2, \cdots, q-1$ and $\alpha\in(0, 1)$, no family of codes with positive rate is list-decodable from $(\gamma_0, \delta_0) = \alpha p_i + (1-\alpha)p_{i+1}$ fraction of insertions and deletions. Note that in \cref{thm:counterexamples} we proved the infeasibility of the vertices of $F_q$ by providing a strategy for the adversary to convert any string into one out of a set of size $O_q(1)$ using the corresponding amount of insertions and deletions. To finish the proof, we similarly present a strategy for the adversary that is obtained by a simple time sharing between the ones used to show infeasibility at $p_i$ and $p_{i+1}$ in \cref{thm:counterexamples}.

Consider a codeword $x\in[q]^n$. As shown in \cref{thm:counterexamples}, the adversary can utilize $n\alpha \cdot p_i$ errors to convert the first $\alpha n$ symbols of $x$ into a string of form $[\sigma_1, \sigma_2, \cdots, \sigma_{i}]^{n\alpha\cdot\frac{i}{q}}$ where $\{\sigma_1, \sigma_2, \cdots, \sigma_{i}\} \subseteq \Sigma$. Similarly, the remaining $n(1-\alpha)p_{i+1}$ errors can be utilized to turn the last $(1-\alpha)n$ symbols of $x$ into a string of the form $[\sigma'_1, \sigma'_2, \cdots, \sigma'_{i+1}]^{n(1-\alpha)\cdot\frac{i+1}{q}}$ where $\{\sigma_1, \sigma_2, \cdots, \sigma_{i+1}\} \subseteq \Sigma$. Note that there are no more than ${q \choose i} i! \cdot {q\choose i+1}(i+1)! = O_q(1)$ of such strings. Therefore, for any given positive rate code, there exists one string of the above-mentioned form which is $(\gamma_0, \delta_0)$-close to exponentially many codewords and, thus, no positive-rate family of codes is list-decodable from $(\gamma_0, \delta_0)$ fraction of insertions and deletions.
\end{proof}
}

\subsection{Feasibility Region: Exact Characterization}

Finally, we will show that the feasibility region is indeed equal to the region $F_q$ described in \cref{thm:QaryUpperbound}. 
The proof closely follows the steps taken for the binary case but is significantly more technical. We first formally define $q$-ary Bukh-Ma codes and show they are list-decodable as long as the error rate lies in $F_q$ and then use the concatenation in \cref{sec:concatenation} to obtain \cref{thm:qaryMain}.

\begin{theorem}\label{thm:LowRateListDecQaryCodes}
For any integer $q\geq 2$, $\eps>0$, and sufficiently large $n$, let $C^q_{n, \eps}$ be the following Bukh-Ma code:
$${C}^q_{n, \eps} = \left\{\left(0^r1^r\cdots q^r\right)^{\frac{n}{qr}}\Big| r=\left(\frac{1}{\eps^4}\right)^k, k< \log_{1/\eps^4} n\right\}.$$

For any $(\gamma, \delta) \in (1-\eps)F_q$ it holds that ${C}^q_{n, \eps}$ is list decodable from any $\delta n$ deletions and $\gamma n$ insertions with a list size of $O\left(\frac{q^5}{\eps^2}\right)$.
\end{theorem}
We remark that in the case of $q=2$, \cref{thm:LowRateListDecQaryCodes} improves over \cref{thm:LowRateListDecBinaryCodes} in terms of the dependence of the list size on $\eps$.

\subsubsection{Proof Sketch for \cref{thm:LowRateListDecQaryCodes}}
To prove \cref{thm:LowRateListDecQaryCodes}, we show that Bukh-Ma codes are list-decodable as long as the error rate $(\gamma, \delta)$ lies beneath the line that connects a pair of consecutive non-zero vertices of $F_q$. 

In other words, for \noSTOC{any pair of points }\STOConly{pairs }$\left(\frac{i(i-1)}{q}, \frac{q-i}{q}\right)$ and $\left(\frac{i(i+1)}{q}, \frac{q-i-1}{q}\right)$ we consider the line passing through them (see \cref{fig:border-characteristics}), i.e., 
\begin{equation}
\gamma + (2i)\delta = \frac{(2q-1)i-i^2}{q}, \quad i=1, \cdots, q-1\label{eqn:line-segments}
\end{equation}
\begin{figure}[]
  \centering
  \noSTOC{\includegraphics[width=.9\textwidth]{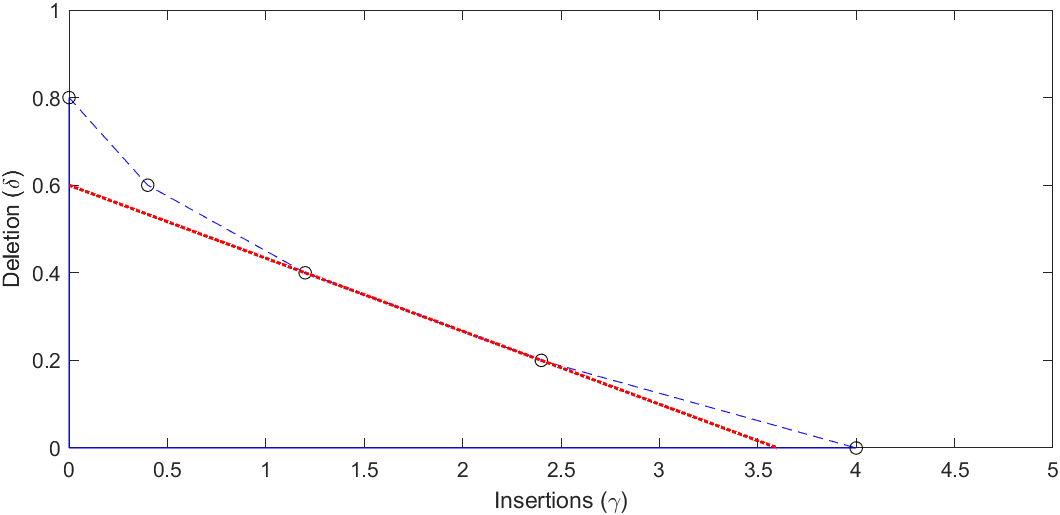}}
  \STOConly{\includegraphics[width=\linewidth]{Characteristic.png}}
  \caption{In the feasibility region for $q=5$, the line passing through $(1.2, 0.4)$ and $(1.8, 0.3)$ (indicated with red dotted line) is characterized as $\gamma + 6\delta \leq 3.6$. (Corresponding to $i=3$ in \cref{eqn:line-segments})}\label{fig:border-characteristics}
\end{figure}
and show that as long as 
$\gamma + (2z)\delta \leq (1-\eps)\frac{(2q-1)z-z^2}{q}$ for some $z\in\{1,\cdots,q-1\}$, Bukh-Ma codes are list-decodable. Note that the union of such areas is equal to $(1-\eps)F_q$.

The analysis for each line follows the arguments for the binary case. Namely, we assume that $k$ codewords can be converted to some center string $v$ via $(\gamma, \delta)$ fraction of errors. Then, using an appropriate advantage notion and considering some coupled statistic processes obtained by sampling substrings, we show that $k$ is bounded above by some $O_q\left(\poly(1/\eps)\right)$.

The only major difference is that the notion of bias cannot be directly used for $q$-ary alphabets. In this general case, instead of keeping track of the variance of the bias, we keep track of the sum of the variances of the frequency of the occurrence of each symbol. We show that this quantity increases by some constant after each substring sampling (analogous to \cref{lem:process-variance-guarantee}) by showing that a positive advantage requires that the frequency of occurrence of at least one of the symbols to be $\eps$-different for two different values of the random variable $T$ (analogous to \cref{lem:polarization}). The rest of this section contains more formal description of generalized notions and proofs for generalized $q$-ary claims.

\subsection{Generalized Notation and \noSTOC{Preliminary Lemmas}\STOConly{Preliminaries}}
To prove \cref{thm:LowRateListDecQaryCodes}, we need to generalize some of the notions and respective preliminary lemmas for the binary case.
We start with defining $i$th order advantage.
\begin{definition}[$i$th order $q$-ary advantage of matching $M$]
For a pair of positive integers $i < q$, a pair of $q$-ary strings $a$ and $b$, and a matching $M$ between $a$ and $b$, we define \emph{$i$th order $q$-ary advantage of $a$ to $b$} as follows:
\noSTOC{$$\adv^{q, i}_M(a, b)=\frac{(2i+1)|M|-|a|-\frac{i+i^2}{q} \cdot |b|}{|a|}$$}
\STOConly{$\adv^{q, i}_M(a, b)=\frac{(2i+1)|M|-|a|-\frac{i+i^2}{q} \cdot |b|}{|a|}.$}
\end{definition}
Note that the notion of advantage utilized for the binary case is obtained for $q=2$ and $i=1$ in the above definition. 
The notions of $i$th order advantage between two strings (that is independent of a specific matching, i.e., $\adv^{q, i}(a, b)$) and infinite $i$th order advantage are defined in a similar manner to the binary case.

\begin{remark}\label{rmk:qAryAdvInterpretation}
In the same spirit as of the binary case, $\adv^{q, i}_M(a, b)$ is simply the value of 
\noSTOC{$$|b|\left(\frac{(2q-1)i-i^2}{q}-(2i)\delta_M-\gamma_M\right)$$}
\STOConly{$|b|\left(\frac{(2q-1)i-i^2}{q}-(2i)\delta_M-\gamma_M\right)$}
normalized by the length of $a$. \noSTOC{Indeed,
\begin{eqnarray*}
\adv^{q, i}_M(a, b) &=& \frac{(2i+1)|M|-|a|-\frac{i+i^2}{q}\cdot|b|}{|a|}\\
&=& \frac{(2i+1)|b|(1-\delta_M)-|b|(1-\delta_M+\gamma_M)-\frac{i+i^2}{q}\cdot|b|}{|a|}\\
&=& \frac{|b|}{|a|}\cdot\left[(2i+1)(1-\delta_M)-(1-\delta_M+\gamma_M)-\frac{i+i^2}{q}\right]\\
&=& \frac{|b|}{|a|}\cdot\left(\frac{(2q-1)i-i^2}{q} -(2i)\delta_M - \gamma_M\right).
\end{eqnarray*}}
\end{remark}

\begin{lemma}\label{lem:qAryLengthPositiveAdv}
If for strings $a$ and $b$, $\adv^{q, i}(a, b) \geq 0$, then $|a|$ and $|b|$ are within a $q$ factor of each other.
\end{lemma}
\STOConly{The proof of this lemma is similar to the binary case and can be found in the extended version of this argument.}
\noSTOC{\begin{proof}
$\adv^{q, i}(a, b) \geq 0$ implies that for some matching $M$, 
\begin{eqnarray}
&&adv^{q, i}_M \geq 0 \Rightarrow |a|+\frac{i+i^2}{q}\cdot|b| \leq (2i+1)|M|\nonumber\\
&\Rightarrow& q|a|+(i+i^2)\cdot|b| \leq q(2i+1)|M| \leq q(2i+1)\min(|a|, |b|)\label{eqn:lengths-advantage}
\end{eqnarray}
Now if $|a|\leq |b|$, \eqref{eqn:lengths-advantage} gives
$$q|a|+(i+i^2) |b| \leq q(2i+1) |a| \Rightarrow |b| \leq \frac{2q}{i+1}\cdot |a| \leq q|a|$$
and if $|b|<|a|$, \eqref{eqn:lengths-advantage} gives
$$q|a|+(i+i^2) |b| \leq q(2i+1) |b| \Rightarrow |a| \leq \frac{2iq+q-i-i^2}{q}\cdot |b| \leq q|b|.$$
\end{proof}}

\begin{definition}[$q$-ary Alternating Strings]
For any positive integer $r$, we define the infinite $q$-ary alternating string of run-length $r$ as $A^q_r = (1^r2^r\cdots q^r)^\infty$ and denote its prefix of length $l$ by $A^q_{r,l} = A^q_r[1,l]$.
\end{definition}

\subsection{Proof of \cref{thm:LowRateListDecQaryCodes}}
As mentioned before, \cref{thm:LowRateListDecQaryCodes} can be restated as follows.
\begin{theorem}[Restatement of \cref{thm:LowRateListDecQaryCodes}]\label{thm:restated}
For any integer $q\geq 2$, $\eps>0$, sufficiently large $n$, and any $z\in\{1, 2, \cdots, q-1\}$, the Bukh-Ma code $C^n_{n, \eps}$ from \cref{thm:LowRateListDecQaryCodes} is list decodable from any $\delta n$ deletions and $\gamma n$ insertions with a list size $O\left(q^5 / \eps^2\right)$ as long as 
$\gamma + (2z)\delta \leq (1-\eps)\frac{(2q-1)z-z^2}{q}$.
\end{theorem}

To prove this restated version, once again, we follow the steps taken for the proof of \cref{thm:LowRateListDecBinaryCodes} and assume for the sake of contradiction that there exists a string $v$ and $k = \Omega\left(\frac{q^5}{\eps^2}\right)$ members of ${C}^q_{n, \eps}$ like $A^q_{r_1, n}, A^q_{r_2, n}, \cdots, A^q_{r_k, n}$, so that each $A^q_{r_i, n}$ can be converted to $v$ with $I_i$ insertions and $D_i$ deletions where $I_i+(2z)D_i \leq (1-\eps)\frac{(2q-1)z-z^2}{q}\cdot n$. We define the indices in a way that $r_1 > r_2 > \cdots > r_k$. Given the definition of $C^q_{n,\eps}$, $r_i \geq \frac{r_{i+1}}{\eps^4}$. 

Given \cref{rmk:qAryAdvInterpretation} and \cref{lem:qAryLengthPositiveAdv}, an argument similar to the one presented in \cref{lem:advantage} shows that for all these codewords, 
$\adv^{q, z}(v, A^q_{r_i, n}) \geq \frac{\eps}{q}$. 

We define the following stochastic processes similar to the binary case.
We split the string $v$ into substrings of size $l_1 = r_1\eps^2$, pick one uniformly at random and denote it by $v_1$. We define random variable $A_1 = \adv^{q, z}(v_1, A^q_{r_1})$ and random variables $F^p_1$ for $p=1, 2, \cdots, q$ as the frequency of the occurrence of symbol $p$ in $v_1$. In other words,
\noSTOC{$$F^p_1 = \frac{\Count_p(v_1)}{|v_1|}.$$}
\STOConly{$F^p_1 = \frac{\Count_p(v_1)}{|v_1|}.$}
We continue this process for $j=2, 3, \cdots, k$ by splitting each $v_{j-1}$ into substrings of length $l_j = r_j\eps^2$, picking $v_j$ uniformly at random, and defining $A_j = \adv^{q, z}(v_j, A^q_{r_j})$ and 
$F^p_j = \frac{\Count_p(v_j)}{|v_j|}$ for all $p\in\{1, 2, \cdots, q\}$.
We then define the sequence of real numbers $f_1 ,f_2, \cdots, f_k$ as follows:
\noSTOC{$$f_i = \sum_{p=1}^q \Var(F^p_i).$$}
\STOConly{$f_i = \sum_{p=1}^q \Var(F^p_i).$}
This series of real numbers will play the role of $\Var(B_i)$ in the binary case. 
\noSTOC{\begin{lemma}
The following hold for $A_1, A_2, \cdots, A_k$ and $F^p_1, F^p_2, \cdots, F^p_k$ for all $p\in\{1, 2, \cdots, q\}$.
\begin{enumerate}
\item $\E[F^p_i] = F^p_{i-1}$
\item $\E[A_i] \geq \frac{\eps}{q}$
\end{enumerate}
\end{lemma}}
\STOConly{\begin{lemma}
The following hold for $A_1, \cdots, A_k$ and $F^p_1, \cdots, F^p_k$ for all $p\in\{1, 2, \cdots, q\}$: $(1) \E[F^p_i] = F^p_{i-1}, \textnormal{and } (2) \E[A_i] \geq \frac{\eps}{q}$.
\end{lemma}}
\begin{proof}
Since $v_i$ is a substring of $v_{i-1}$ chosen uniformly at random, the overall frequency of symbol $p$ is equal to the average frequency of its occurrence in each substrings. The second item can be derived as in \cref{lem:process-guarantees}.
\end{proof}

The next lemma mimics \cref{lem:process-variance-guarantee} for the binary case.

\begin{lemma}\label{lem:qAry-process-variance-guarantee}
For the sequence $f_1, f_2, \cdots, f_k$, we have that
\noSTOC{$$f_{i+1} \geq f_i + \Omega\left(\frac{\eps^2}{q^4}\right).$$}
\STOConly{$f_{i+1} \geq f_i + \Omega(\eps^2/q^4).$}
\end{lemma}

Using \cref{lem:qAry-process-variance-guarantee}, \cref{thm:restated} can be simply proved as follows.

\begin{proof}[{\bf Proof of \cref{thm:restated}}]
Note that each $f_i$ is the summation of the variance of $q$ random variables that take values in $[0, 1]$. Therefore, their value cannot exceed $q$.
Since $f_{i+1} \geq f_i + \Omega(\eps^2/q^4)$, the total length of the series, $k$, may not exceed $O\left(\frac{q^5}{\eps^2}\right)$.
This implies that the list size is $O\left(\frac{q^5}{\eps^2}\right)$.
\end{proof}

We now present the proof of \cref{lem:qAry-process-variance-guarantee}.

\begin{proof}[{\bf Proof of \cref{lem:qAry-process-variance-guarantee}}]
To relate $f_i$ and $f_{i+1}$, we utilize the law of total variance as follows:
\begin{eqnarray}
\Var(F^p_{i+1})&=&\Var\left(\E[F^p_{i+1}|v_{i}]\right) + \E\left[\Var(F^p_{i+1}|v_{i})\right]\nonumber\allowdisplaybreaks\\
&=&\Var\left(F^p_{i}\right) + \E\left[\Var(F^p_{i+1}|v_{i})\right]\label{eqn:qArytotalVar1}
\end{eqnarray}
Equation \eqref{eqn:qArytotalVar1} comes from the fact that the average frequency of symbol $p$ in substrings of length $l_{i+1}$ of $v_i$ is equal to the frequency of $p$ in $v_i$. Having this, we see that it suffices to show that $\E\left[\Var(F^p_{i+1}|v_{i})\right] \geq \Omega\left(\eps^2/q^4\right)$. Similar to \cref{lem:process-variance-guarantee} we define $E_j$ for $j=1,2,\cdots,q$ and $E_e$ respectively as the event that the projection of $v_{i+1}$ falls inside a $j^{r_{i+1}}$ in $A_{r_{i+1}}$ or a string containing multiple symbols. We also define the random variable $T$ out of $\{e, 1, 2, \cdots, q\}$ that indicates which one of these events is realized. Once again, we use the law of total variance to bound $\E\left[\Var(F^p_{i+1}|v_i)\right]$.
\noSTOC{\begin{eqnarray}
\E\left[\Var(F^p_{i+1}|v_{i})\right] &=&\E_{v_i}\Bigg[\Var_{T}(\E\left[F^p_{i+1}|v_i, T\right]) + \E_{T}\left[\Var(F^p_{i+1}|v_i, T)\right]\big]\nonumber\\
&\geq& \E_{v_i}\big[\Var_{T}(\E\left[F^p_{i+1}|v_i, T\right])\Bigg]\label{eqn:qAryVarLoweboundStep2}
\end{eqnarray}}
\STOConly{\begin{align}
\E\left[\Var(F^p_{i+1}|v_{i})\right] &=\E_{v_i}\Big[\Var_{T}\left(\E\left[F^p_{i+1}|v_i, T\right]\right)\nonumber\\
 &+ \E_{T}\big[\Var(F^p_{i+1}|v_i, T)\big]\Big]\nonumber\\
&\geq \E_{v_i}\left[\Var_{T}\left(\E\left[F^p_{i+1}|v_i, T\right]\right)\right]\label{eqn:qAryVarLoweboundStep2}
\end{align}}
Combining \eqref{eqn:qArytotalVar1} and \eqref{eqn:qAryVarLoweboundStep2} gives
\noSTOC{\begin{eqnarray}
&&\Var(F^p_{i+1}) \geq \Var\left(F^p_{i}\right) + \E_{v_i}\left[\Var_{T}\left(\E\left[F^p_{i+1}|v_i, T\right]\right)\right]\nonumber\\
&\Rightarrow& \sum_{p=1}^q \Var(F^p_{i+1}) \geq \sum_{p=1}^q \Var\left(F^p_{i}\right) + \sum_{p=1}^q \E_{v_i}\left[\Var_{T}\left(\E\left[F^p_{i+1}|v_i, T\right]\right)\right]\nonumber\allowdisplaybreaks\\
&\Rightarrow& f_{i+1} \geq f_{i} + \sum_{p=1}^q \E_{v_i}\left[\Var_{T}\left(\E\left[F^p_{i+1}|v_i, T\right]\right)\right]\nonumber\\
&\Rightarrow& f_{i+1} \geq f_{i} + \E_{v_i}\left[\sum_{p=1}^q \Var_{T}\left(\E\left[F^p_{i+1}|v_i, T\right]\right)\right]\label{eqn:qAry-incremental-increase}
\end{eqnarray}}
\STOConly{
\begin{align}
&\Var(F^p_{i+1}) \geq \Var\left(F^p_{i}\right) + \E_{v_i}\left[\Var_{T}\left(\E\left[F^p_{i+1}|v_i, T\right]\right)\right]\nonumber\\
&\Rightarrow \sum_{p=1}^q \Var(F^p_{i+1}) \geq \sum_{p=1}^q \Var\left(F^p_{i}\right) + \sum_{p=1}^q \E_{v_i}\left[\Var_{T}\left(\E\left[F^p_{i+1}|v_i, T\right]\right)\right]\nonumber\allowdisplaybreaks\\
&\Rightarrow f_{i+1} \geq f_{i} + \E_{v_i}\left[\sum_{p=1}^q \Var_{T}\left(\E\left[F^p_{i+1}|v_i, T\right]\right)\right]\label{eqn:qAry-incremental-increase}
\end{align}}

Note that the term 
\noSTOC{$\Var_{T}\left(\E\left[F^p_{i+1}|v_i, T\right]\right)$ }
\STOConly{$\Var_{T}(\E[F^p_{i+1}|v_i, T])$ }
refers to the variance of a ($q+1$)-valued random variable that takes the value 
\noSTOC{$\E_{v_i}\left[F^p_{i+1}|v_i, T=t\right]$ }
\STOConly{$\E_{v_i}[F^p_{i+1}|v_i, T=t]$ }
with probability $\Pr\{T=t|v_i\}$ for $t\in\{e, 1, 2, \cdots, q\}$. Once again, we present a crucial lemma that bounds from below the sum of variances of frequencies with respect to $T$ assuming that the overall advantage is large enough.
\begin{lemma}\label{lem:qAry-polarization}
For any realization of $v_i$, the following holds true if $\adv^{q, z}(v_i, A_{r_{i+1}}) \geq 3q\eps^2$:
$$ \sum_{p=1}^q \Var_{T}\left(\E\left[F^p_{i+1}|v_i, T\right]\right) \geq \left(\frac{\adv^{q, z}(v_i, A_{r_{i+1}}) - 3q\eps^2}{2z+1}\right)^2$$
\end{lemma}
We defer the proof of \cref{lem:qAry-polarization} to \cref{sec:qAry-polarization}. Using Jensen inequality, the fact that $z\leq q$, and \cref{lem:qAry-polarization} along with \eqref{eqn:qAry-incremental-increase} give that
\noSTOC{$$f_{i+1} \geq f_i + \E_{v_i}\left[
\left(\frac{\adv^{q, z}(v_i, A_{r_{i+1}}) - 3q\eps^2}{2z+1}\right)^2
\right] \noSTOC{\geq f_i + 
\left(\frac{\eps/q - 3q\eps^2}{2q+1}\right)^2} = f_i + \Omega\left(\frac{\eps^2}{q^4}\right)
$$}
\STOConly{$f_{i+1} \geq f_i + \E_{v_i}\left[
\left(\frac{\adv^{q, z}(v_i, A_{r_{i+1}}) - 3q\eps^2}{2z+1}\right)^2
\right] \noSTOC{\geq f_i + 
\left(\frac{\eps/q - 3q\eps^2}{2q+1}\right)^2} = f_i + \Omega\left(\frac{\eps^2}{q^4}\right)
$}
for sufficiently small $\eps > 0$.
\end{proof}

\subsection{Proof of \cref{thm:qaryMain}}
To establish \cref{thm:qaryMain}, we closely follow the concatenation scheme presented in \cref{sec:concatenation}. 
In the following, we provide a high-level description of the proof skipping the details mentioned in \cref{sec:concatenation} and highlighting the necessary extra steps.

The construction of the concatenated code is exactly as in \cref{sec:concatenation} with the exception that the inner code is defined over an alphabet of size $q$. 
Note that if $(\gamma, \delta)\in (1-\eps) F_q$, then $(\gamma, \delta)$ lies underneath one of the lines in the set of lines represented by \eqref{eqn:line-segments}. In other words, there exists some $z\in\{1, 2, \cdots, q-1\}$ for which
\noSTOC{$$\gamma + (2z)\delta \leq (1-\eps)\left(\frac{(2q-1)z-z^2}{q}\right).$$}
\STOConly{$\gamma + (2z)\delta \leq (1-\eps)\left(\frac{(2q-1)z-z^2}{q}\right).$}
Similar to \cref{sec:concatenation}, we define the notion of {\em error count} for each block in the codewords of the concatenated code as 
\noSTOC{$$ (I+2z\cdot D) \cdot \frac{q}{(2q-1)z-z^2} $$}
\STOConly{$ (I+2z\cdot D) \cdot \frac{q}{(2q-1)z-z^2} $}
where $D$ and $I$ denote the number of deletions and insertions occurred in the block respectively. As in \cref{sec:concatenation} one can show that at least $\frac{3\eps}{4}\cdot n_\textout$ of the blocks contain no more than $\left(1-\frac{\eps}{4}\right) n_\textin$ error count. We denote the set of all such blocks by $S$. Once again, we partition $S$ into subsets $S_1, S_2, \cdots$ depending on the number of deletions occurred in the set. More precisely, we define $S_i \subseteq S$ as the set of blocks in $S$ that contain a number of deletions that is in the range $\left[n_{\text{in}}\cdot\frac{\eps}{16q}\cdot (i-1), n_{\text{in}}\cdot \frac{\eps}{16q}\cdot i\right)$ for $i=1, 2, \cdots, 16q/\eps$. 
Once again, the following hold true:
\noSTOC{\begin{enumerate}}
\STOConly{\begin{enumerate}[leftmargin=5mm]}
    \item We have that $S=\dot\bigcup_{i=1}^{16q/\eps} S_i$. By the Pigeonhole principle, for some $i^*\in\left[1, 16q/\eps\right]$, $|S_{i^*}|\ge \frac{3\eps^2}{64q}n_{\text{out}}$.
    \item Take some $i\in\{1, 2, \cdots, 16q/\eps\}$ and some block in $S_{i}$. Say $D$ deletions have occurred in that block. Then, the total number of insertions is at most 
    $(1-\eps/4)\frac{(2q-1)z-z^2}{q} n_\textin -2z D$.
    Therefore, the total length of the block is  
\STOConly{    
	$n_\textin - D (1-\eps/4)\frac{(2q-1)z-z^2}{q} n_\textin -2z D$
	\begin{eqnarray}
    &=& n_\textin\cdot\left[
    1+\left(1-\frac{\eps}{4}\right)\frac{(2q-1)z-z^2}{q}
    \right] - (2z+1)D\label{eq:block-length}
    \end{eqnarray}}
\noSTOC{    \begin{eqnarray}
    && n_\textin - D (1-\eps/4)\frac{(2q-1)z-z^2}{q} n_\textin -2z D\nonumber\\
    &=& n_\textin\cdot\left[
    1+\left(1-\frac{\eps}{4}\right)\frac{(2q-1)z-z^2}{q}
    \right] - (2z+1)D\label{eq:block-length}
    \end{eqnarray}}
    which is no more than 
    \begin{eqnarray}
    n_\textin\cdot\left[
    1+\left(1-\frac{\eps}{4}\right)\frac{(2q-1)z-z^2}{q}
     - \frac{\eps}{16q}(i - 1)(2z+1)
     \right]\label{eqn:block-length-upper-bound}
    \end{eqnarray}

Based on these observations, it is easy to verify that the decoding algorithm and analysis as presented in \cref{sec:concatenation} and \cref{alg:conctenatedDecoder} work for the $q$-ary case with the following minor modifications:
\noSTOC{\begin{enumerate}}
\STOConly{\begin{enumerate}[leftmargin=3mm]}
    \item Based on \eqref{eqn:block-length-upper-bound}, the parameter $w$ determining the length of the window should be 
    \begin{equation}
    w = \left\lfloor\frac{n_\textin\cdot\left[
    1+\left(1-\frac{\eps}{4}\right)\frac{(2q-1)z-z^2}{q}
     - \frac{\eps}{16q}(i - 1)(2z+1)
     \right]}{n_\textin\eps/16} \right\rfloor+1.\label{eqn:window-size}
    \end{equation}
    \item As in \eqref{eqn:eps-in-choice}, parameter $\eps_\textin$ has to be chosen such that the error count in decoding windows does not exceed $n_\textin(1-\eps_\textin)$. Note that the choice of shifting steps for the decoding window from \eqref{eqn:window-size} may add up to $\frac{n_\textin \eps}{16}$ additional insertions to the decoding window. Further, there is up to $n_\textin \frac{\eps}{16q}$ uncertainty in the total length of the block from \eqref{eq:block-length} since $D\in \left[n_{\text{in}}\cdot\frac{\eps}{16q}\cdot (i-1), n_{\text{in}}\cdot \frac{\eps}{16q}\cdot i\right)$. This can also add up to $n_\textin \frac{\eps}{16q}(2z+1) \leq \frac{\eps}{8}$ insertions. Therefore, we need
\noSTOC{\begin{eqnarray*}
    n_\textin(1-\eps/4)+n_\textin\left(\frac{\eps}{16} + \frac{\eps}{8}\right)\cdot \frac{q}{(2q-1)z-z^2} \leq n_\textin(1-\eps_\textin).
    \end{eqnarray*}}
\STOConly{$n_\textin(1-\eps/4)+n_\textin\left(\frac{\eps}{16} + \frac{\eps}{8}\right)\cdot \frac{q}{(2q-1)z-z^2} \leq n_\textin(1-\eps_\textin).$}
    Note that $\frac{q}{(2q-1)z-z^2} \leq \frac{q}{2q-2}\leq 1$. Hence, it suffuces that
    $1-\frac{\eps}{4}+\frac{\eps}{8}+\frac{\eps}{16} \leq 1-\eps_\textin$ or equivalently, $\eps_\textin\leq \frac{\eps}{16}$.

    \item Some modifications are necessary to the parameters of the outer code. Notably, for alphabet size $q$, $|S_{i^*}|\ge \frac{3\eps^2}{64q}n_{\text{out}}$ and the fraction of deletions can be as high as $1-\frac{1}{q}$. This requires $\delta_{\text{out}}=1-\frac{3\eps^2}{128q^2}$. 
    \item Finally, note the the value of $z$ is not know to the decoder. So the decoder has to run the algorithm with modifications mentioned above for all possible values of $z=1, 2, \cdots, q-1$ and the output the union of all lists produced.
\end{enumerate}
\end{enumerate}

\subsection{Proof of \cref{lem:qAry-polarization}}\label{sec:qAry-polarization}
\pushQED{\qed}
We break down this proof into four steps. In the first step, similar to \cref{lem:polarization}, we modify $v_i$ and $A_{r_{i+1}, n}$ into a simpler structure without significantly changing the advantage. In the second step, we provide an upper bound for the advantage in this modified version that depends on the local frequencies of symbols, more specifically, on what we refer to as $\E\left[F^j_{i+1}|v_i, T=j\right]$. In Step 3, we show that these upper-bounds would yield a non-positive value on the advantage if one replaces the local frequencies with the overall frequency of symbols in $v_i$, i.e., $F^j_i$. In the fourth and last step, we show that this means that the local frequencies have to significantly deviate from global ones to attain the advantage achieved by $\bar{M}_i$ (i.e., $\adv^{q, z}_{\bar{M}_i}$), so much that the lower-bound promised in the lemma's statement is achieved.

\paragraph{Step 1. Modifying $v_i$ and $A_{r_{i+1}, n}$ for the sake of simplicity:}
The proof starts with modifying $v_i$, $A_{r_{i+1}, n}$, and the advantage-yielding matching $M_i$ between them in a way that only slightly changes the value of advantage taking steps identical to the one in \cref{lem:polarization}.
Similar to \cref{lem:polarization}, we denote the projection of $v_i$ under $M_i$ by $g = v_i\rightarrow M_i$. (See \cref{fig:transformtion} for a depiction of the steps in binary case.)

\begin{enumerate}
\item \label{step:qAryModificationStepOne}First, we delete all substrings of $U_e$--i.e., substrings of length $l_{i+1}$ in $v_i$ whose projection does not entirely fall into some stretch of $j^{r_{i+1}}$--from $v_i$.
\item We reorder the substrings of length $l_{i+1}$ in $v_i$ by shifting all $U_j$ substrings together and the projections in $g$ to preserve the remainder of $M_i$ from step \ref{step:qAryModificationStepOne}.
\item At this point, string $g$ consists of a stretch of symbol $1$ followed by a stretch of symbol $2$, etc. If the length of all stretches are not equal, we add adequate symbols to each stretch to make $g$ have the form of $1^t 2^t \cdots q^t$.
\end{enumerate}

To track the changes in $\adv^{q, z}_{M_i}$ during this transformation, we track how $|M_i|$, $|v_i|$ and $|g|$ change throughout the three steps mentioned above.

In the first step, a total of up to 
$|U_e| l_{i+1}$
elements are removed from $v_i$ and $M_i$. Note that since the run length of $A_{r_{i+1}}$ is $r_{i+1}$, there can only be $\frac{|g|}{r_{i+1}}$ substrings in $U_e$. Therefore,
\noSTOC{$$|U_e| l_{i+1} \leq \frac{|g|l_{i+1}}{r_{i+1}} = |g|\eps^2 \leq 2\eps^2|v_i|.$$}\STOConly{$|U_e| l_{i+1} \leq \frac{|g|l_{i+1}}{r_{i+1}} = |g|\eps^2 \leq 2\eps^2|v_i|.$}

The second step preserves $|M_i|$, $|v_i|$ and $|g|$.

Finally, since $g$ is a substring of $A_{r_{i+1}}$, the third step increases $|g|$ only by up to $q r_{i+1}$. Note the run length of the $A_{r_{i+1}}$s and consequently $l_{i+1}$s are different by a multiplicative factor of at least $\frac{1}{\eps^4}$ by the definition of the code $\mathcal{C}$. Therefore,
$qr_{i+1} = \frac{ql_{i+1}}{\eps^2} = \frac{q l_{i+1}|v_i|}{\eps^2|v_i|} = \frac{q l_{i+1}|v_i|}{\eps^2 l_{i}}
\leq \eps^2q|v_i|$. 

Overall, the value of the 
$\adv^{q, z}_{M_i} = \frac{(2z+1)|M|-|v_i|-\frac{z+z^2}{q} \cdot |g|}{|v_i|}$
can be affected by a maximum of $2z\times2\eps^2|v_i| + q\eps^2|v_i| = (2z+q)\eps^2|v_i| \leq 3q \eps^2|v_i|$ decrease in the numerator and $\eps^2|v_i|$ decrease in the denominator. Therefore, the eventual advantage does not drop below $\adv^{q, z}_{M_i} - 3q\eps^2$.
Let us denote the transformed versions of $v_i$, $g$, and $M_i$ by $\bar{v}_i$, $\bar{g}$, and $\bar{M}_i$ respectively. We have shown that
\begin{equation}
\adv^{q, z}_{\bar{M}_i} \geq \adv^{q, z}_{M_i} - 3q\eps^2.\label{eqn:qAryAdvLowerBoundAfterTransformation}
\end{equation}

\paragraph{Step 2. Bounding Above $\adv^{q, z}_{\bar{M}_i}$ with $f^*$:}
Let $\bar{v}_i = (\bar{v}_i^1, \bar{v}_i^2, \cdots, \bar{v}_i^q)$ so that $\bar{v}_i^j$ corresponds to the part of $\bar{v}_i$ that is mapped to $j^t$ under $\bar{M}_i$. Further, let $f^*_j = \E\left[F^j_{i+1}|v_i, T=j\right]$ represent the frequency of the occurrence of symbol $j$ in $\bar{v}_i^j$ as a shorthand, i.e.,
\noSTOC{$$f^*_j = \frac{\Count_j(\bar{v}_i^j)}{|\bar{v}_i^j|}$$}
\STOConly{$f^*_j = \frac{\Count_j(\bar{v}_i^j)}{|\bar{v}_i^j|}$}
and $p_j$ be the relative length of $\bar{v}_i^j$, i.e., \noSTOC{$$p_j=\frac{|\bar{v}_i^j|}{|\bar{v}_i|}.$$}\STOConly{$p_j=\frac{|\bar{v}_i^j|}{|\bar{v}_i|}.$}
In this section, we compute an upperbound for $\adv^{q, z}_{\bar{M}_i}$ that depends on $f^*_j$s. For the sake of simplicity, from now on we assume, without loss of generality, that 
\noSTOC{$$\Count_1(\bar{v}_i^1) \geq  \Count_2(\bar{v}_i^2) \geq \cdots \geq \Count_q(\bar{v}_i^q)$$}
\STOConly{$\Count_1(\bar{v}_i^1) \geq  \Count_2(\bar{v}_i^2) \geq \cdots \geq \Count_q(\bar{v}_i^q)$}
or equivalently,
\noSTOC{$$f^*_1 p_1 \geq f^*_2 p_2 \geq \cdots \geq f^*_q p_q.$$}
\STOConly{$f^*_1 p_1 \geq f^*_2 p_2 \geq \cdots \geq f^*_q p_q.$}

Consider the matching between $\bar{v}_i$ and $\bar{p}$ that, for any $j\in\{1, 2, \cdots, q\}$ matches as many $j$s as possible from $j^t$ to $\bar{v}^j_i$. This matching clearly yields the largest possible advantage between the two that is an upperbound for the $\adv^{q, z}_{\bar{M_i}}$. Similar to the binary case, we find a $t$ that maximizes this advantage and use its advantage as an upper-bound for $\adv^{q, z}_{\bar{M_i}}$.

Let $c$ be so that 
$f^*_c |\bar{v}_i^c| > t \geq f^*_{c+1} |\bar{v}_i^{c+1}|$
. Then, increasing $t$ by one would increase the length of $\bar{p}$ by $q$ and increases the size of the matching by $c$. To see the effect of this increment on the advantage, note that the denominator does not change and the numerator changes by $c(2z+1) - \frac{z+z^2}{q}\cdot q$. This change in advantage is positive as long as
\noSTOC{\begin{eqnarray*}
&&c(2z+1) - (z+z^2) \geq 0\\
&\Leftrightarrow& c \geq \frac{z+z^2}{2z+1} = \frac{z}{2} + \left(\frac{1}{4} - \frac{1}{4(2z+1)}\right).
\end{eqnarray*}}
\STOConly{$c(2z+1) - (z+z^2) \geq 0
\Leftrightarrow c \geq \frac{z+z^2}{2z+1} = \frac{z}{2} + \left(\frac{1}{4} - \frac{1}{4(2z+1)}\right).$}
Note that the term $\frac{1}{4} - \frac{1}{4(2z+1)}$ is always between $\left[0, \frac{1}{4}\right]$. Hence, incrementing $t$ increases the advantage as long as $c \geq \lfloor\frac{z}{2}\rfloor + 1$. This means that the highest possible advantage is derived when $t = f^*_w |\bar{v}_i^w|$ for $w = \lfloor\frac{z}{2}\rfloor + 1$. With this value for $t$, the matching contains $f^*_j |\bar{v}_i^j|$ edges between $j^t$ and $|\bar{v}_i^j|$ for all $j> w$ and $t$ edges between $j^t$ and $|\bar{v}_i^j|$ for $j\leq w$. Therefore, the size of this matching is 
\noSTOC{$$t w + \sum_{j=w+1}^q f^*_j |\bar{v}_i^j|.$$}\STOConly{$t w + \sum_{j=w+1}^q f^*_j |\bar{v}_i^j|.$}
This yields the following advantage
\begin{eqnarray*}
&&\frac{(2z+1)
\left[t w + \sum_{j=w+1}^q f^*_j |\bar{v}_i^j|\right]
-|\bar{v}_i|-\frac{z+z^2}{q} \cdot qt}{|\bar{v}_i|}\\
\noSTOC{&=&\frac{(2z+1)
\left[f^*_w |\bar{v}_i^w|w + \sum_{j=w+1}^q f^*_j |\bar{v}_i^j|\right]
-|\bar{v}_i|-\frac{z+z^2}{q} \cdot qf^*_w |\bar{v}_i^w|}{|\bar{v}_i|}\\}
&=&(2z+1)
\left[f^*_w p_ww + \sum_{j=w+1}^q f^*_j p_j\right]
- 1 -(z+z^2) \cdot f^*_w p_w\allowdisplaybreaks\\
&=&
\left[(2z+1)w -(z+z^2)\right] \cdot f^*_w p_w + (2z+1)\sum_{j=w+1}^q f^*_j p_j
- 1 
\end{eqnarray*}
We remind that this is an upper-bound on the $\adv^{q, z}_{\bar{M}_i}$. Next, we plug in $w = \lfloor\frac{z}{2}\rfloor + 1$ into this bound. Note that
\noSTOC{$$(2z+1)w -(z+z^2) = z(2w - z) + w - z = 
\left\{\begin{matrix}
    \frac{3z+2}{2} & \textnormal{If $z$ is even} \\
    \frac{z+1}{2} & \textnormal{If $z$ is odd}
  \end{matrix}\right.$$}
\STOConly{$(2z+1)w -(z+z^2) = z(2w - z) + w - z$ which is equal to $\frac{3z+2}{2}$ if $z$ is even and $\frac{z+1}{2}$ if $z$ is odd.}

Therefore, we have the following set of upper-bounds on the advantage 
\noSTOC{\begin{eqnarray}
&\adv^{q, z}_{\bar{M}_i} \leq \frac{3z+2}{2} \cdot f^*_w p_w + (2z+1)\sum_{j=w+1}^q f^*_j p_j - 1 & \textnormal{If $z$ is even} \label{eqn:EvenIBound}\\
&\adv^{q, z}_{\bar{M}_i} \leq \frac{z+1}{2} \cdot f^*_w p_w + (2z+1)\sum_{j=w+1}^q f^*_j p_j - 1 & \textnormal{If $z$ is odd} \label{eqn:OddIBound}
\end{eqnarray}}
\STOConly{\begin{eqnarray}
\adv^{q, z}_{\bar{M}_i} \leq \frac{3z+2}{2} \cdot f^*_w p_w + (2z+1)\sum_{j=w+1}^q f^*_j p_j - 1, \textnormal{If $z$ is even} \label{eqn:EvenIBound}\allowdisplaybreaks\\
\adv^{q, z}_{\bar{M}_i} \leq \frac{z+1}{2} \cdot f^*_w p_w + (2z+1)\sum_{j=w+1}^q f^*_j p_j - 1, \textnormal{If $z$ is odd} \label{eqn:OddIBound}
\end{eqnarray}}

\paragraph{Step 3. Proving Non-positivity of the Bound from Step 3 for Unit Sum Vectors:}
In this step, we show that the bounds \eqref{eqn:EvenIBound} and \eqref{eqn:OddIBound} on advantage that were presented in Step 2 are necessarily non-positive for any vector $(f^*_1, \cdots, f^*_q)$ with unit sum including the vector of overall frequencies $\bar{f} = (\bar{f}_1, \cdots, \bar{f}_q)$ where $\bar{f}_j = \frac{\Count_j(\bar{v}_i)}{|\bar{v}_i|}=F^j_i$. In Step 4, we use this fact to show that $f^*$ needs to deviate noticeably from $\bar{f}$ which gives that the variance of frequencies with respect to $T$ is large enough, thus finishing the proof.

\begin{proposition}\label{thm:non-positivity}
Let $(p_1, \cdots, p_q)$ and $(f^*_1, \cdots, f^*_q)$ be two positive real vectors with unit sum that satisfy 
\noSTOC{$$f^*_1 p_1 \geq f^*_2 p_2 \geq \cdots \geq f^*_q p_q.$$}
\STOConly{$f^*_1 p_1 \geq f^*_2 p_2 \geq \cdots \geq f^*_q p_q.$}
Then, for all integers $1\leq z<q$, the following hold for $w = \lfloor\frac{z}{2}\rfloor + 1$:
\begin{enumerate}
    \item If $z$ is even, 
    \noSTOC{$$\frac{3z+2}{2} \cdot f^*_w p_w + (2z+1)\sum_{j=w+1}^q f^*_j p_j \leq 1.$$}\STOConly{$\frac{3z+2}{2} \cdot f^*_w p_w + (2z+1)\sum_{j=w+1}^q f^*_j p_j \leq 1.$}
    \item If $z$ is odd, 
    \noSTOC{$$\frac{z+1}{2} \cdot f^*_w p_w + (2z+1)\sum_{j=w+1}^q f^*_j p_j \leq 1.$$}\STOConly{$\frac{z+1}{2} \cdot f^*_w p_w + (2z+1)\sum_{j=w+1}^q f^*_j p_j \leq 1.$}
\end{enumerate}
\end{proposition}
\noSTOC{We defer the proof of \cref{thm:non-positivity} to \cref{sec:proof-of-non-positivity}.}\STOConly{The proof of \cref{thm:non-positivity} can be found in the extended version of this article.}

\paragraph{Step 4. Large Deviation of $f^*$s from $\bar{f}$s and Large Variance:}
Here we finish the proof assuming $z$ is odd. The even case can be proved in the same way. Note that \cref{thm:non-positivity} gives that for the overall frequency vector $\bar{f}$ which has a unit sum,
\begin{equation}
\frac{z+1}{2} \cdot \bar{f}_w p_w + (2z+1)\sum_{j=w+1}^q \bar{f}_j p_j - 1 \leq 0.\label{eqn:uniformAdvantage}
\end{equation}

However, \eqref{eqn:qAryAdvLowerBoundAfterTransformation} and \eqref{eqn:OddIBound} imply that for local frequency vector $f^*$
\begin{equation}
\frac{z+1}{2} \cdot f^*_w p_w + (2z+1)\sum_{j=w+1}^q f^*_j p_j - 1 \geq \adv^{q, z}_{M_i} - 3q\eps^2.\label{eqn:nonUniformAdvantage}
\end{equation}

Subtracting \eqref{eqn:uniformAdvantage} from \eqref{eqn:nonUniformAdvantage} gives that
\noSTOC{\begin{eqnarray*}
&&\frac{z+1}{2} \cdot p_w (f^*_w - \bar{f}_w) + (2z+1)\sum_{j=w+1}^q (f^*_j - \bar{f}_j) p_j \geq \adv^{q, z}_{M_i} - 3q\eps^2.\\
&\Rightarrow&\frac{z+1}{2} \cdot p_w |f^*_w - \bar{f}_w| + (2z+1)\sum_{j=w+1}^q |f^*_j - \bar{f}_j| p_j \geq \adv^{q, z}_{M_i} - 3q\eps^2.\\
&\Rightarrow&(2z+1)\sum_{j=w}^q |f^*_j - \bar{f}_j| p_j \geq \adv^{q, z}_{M_i} - 3q\eps^2.\allowdisplaybreaks\\
&\Rightarrow&\sum_{j=w}^q |f^*_j - \bar{f}_j| p_j \geq \frac{\adv^{q, z}_{M_i} - 3q\eps^2}{2z+1}.
\end{eqnarray*}}
\STOConly{\begin{align*}
&\frac{z+1}{2} \cdot p_w (f^*_w - \bar{f}_w) + (2z+1)\sum_{j=w+1}^q (f^*_j - \bar{f}_j) p_j \geq \adv^{q, z}_{M_i} - 3q\eps^2.\\
&\Rightarrow\frac{z+1}{2} \cdot p_w |f^*_w - \bar{f}_w| + (2z+1)\sum_{j=w+1}^q |f^*_j - \bar{f}_j| p_j \geq \adv^{q, z}_{M_i} - 3q\eps^2.\\
&\Rightarrow(2z+1)\sum_{j=w}^q |f^*_j - \bar{f}_j| p_j \geq \adv^{q, z}_{M_i} - 3q\eps^2.\allowdisplaybreaks\\
&\Rightarrow\sum_{j=w}^q |f^*_j - \bar{f}_j| p_j \geq \frac{\adv^{q, z}_{M_i} - 3q\eps^2}{2z+1}.
\end{align*}}
This means that there exists some $j_0$ for which 
\noSTOC{$$|f^*_{j_0} - \bar{f}_{j_0}| p_{j_0} \geq \frac{\adv^{q, z}_{M_i} - 3q\eps^2}{2z+1}
\Rightarrow
(f^*_{j_0} - \bar{f}_{j_0})^2 p_{j_0} \geq 
(f^*_{j_0} - \bar{f}_{j_0})^2 p_{j_0}^2 \geq 
\left(\frac{\adv^{q, z}_{M_i} - 3q\eps^2}{2z+1}\right)^2
.$$}
\STOConly{$|f^*_{j_0} - \bar{f}_{j_0}| p_{j_0} \geq \frac{\adv^{q, z}_{M_i} - 3q\eps^2}{2z+1}
\Rightarrow
(f^*_{j_0} - \bar{f}_{j_0})^2 p_{j_0} \geq 
(f^*_{j_0} - \bar{f}_{j_0})^2 p_{j_0}^2 \geq 
\left(\frac{\adv^{q, z}_{M_i} - 3q\eps^2}{2z+1}\right)^2
.$}

Note that 
\noSTOC{\begin{eqnarray*}
\sum_{p=1}^q \Var_{T}\left(\E\left[F^p_{i+1}|v_i, T\right]\right) &=& 
\sum_{p=1}^q \sum_{j=1}^q \left(\E\left[F^p_{i+1}|v_i, T=j\right] - F^p_i\right)^2\Pr\{T=j|v_i\}\\
&\geq&\left(\E\left[F^{j_0}_{i+1}|v_i, T=j_0\right] - F^{j_0}_i\right)^2\Pr\{T=j_0|v_i\}\\
&=&(f^*_{j_0} - \bar{f}_{j_0})^2 p_{j_0} \geq \left(\frac{\adv^{q, z}_{M_i} - 3q\eps^2}{2z+1}\right)^2.
\end{eqnarray*}}
\STOConly{\begin{align*}
&\sum_{p=1}^q \Var_{T}\left(\E\left[F^p_{i+1}|v_i, T\right]\right)\\ 
&\qquad= \sum_{p=1}^q \sum_{j=1}^q \left(\E\left[F^p_{i+1}|v_i, T=j\right] - F^p_i\right)^2\Pr\{T=j|v_i\}\allowdisplaybreaks\\
&\qquad\geq\left(\E\left[F^{j_0}_{i+1}|v_i, T=j_0\right] - F^{j_0}_i\right)^2\Pr\{T=j_0|v_i\}\allowdisplaybreaks\\
&\qquad=(f^*_{j_0} - \bar{f}_{j_0})^2 p_{j_0} \geq \left(\frac{\adv^{q, z}_{M_i} - 3q\eps^2}{2z+1}\right)^2.
\end{align*}}
\qedhere
\popQED


\noSTOC{
\shortOnly{
\newpage
\appendix
\section{Missing Proofs}
}
\subsection{Proof of \cref{thm:non-positivity}}\label{sec:proof-of-non-positivity}
\pushQED{\qed}
To prove \cref{thm:non-positivity} we provide several observations that simplify the form of the solution that yields the maximum value by reducing the number of important free variables.

\begin{observation}\label{obs:solution-format}
Any solution that maximizes the left-hand-side satisfies
$$f^*_1p_1 = f^*_2p_2=  \cdots = f^*_wp_w.$$
\end{observation}

We start with $f^*_{w-1}p_{w-1} = f^*_wp_w$. Assume by contradiction that $f^*_{w-1}p_{w-1} > f^*_wp_w$. Then, there exists a small positive value $\epsilon$ for which decreasing $f^*_{w-1}$ by $\epsilon$ and increasing $f^*_w$ by $\epsilon$ would preserve $f^*_{w-1}p_{w-1} \geq f^*_wp_w$ but increase the overall value of the expression. This contradicts the fact that the solution maximizes the left-hand-side value. Similarly, if 
$f^*_{w-2}p_{w-2} >f^*_{w-1}p_{w-1} = f^*_wp_w$, same idea executed on $f^*_{w-2}p_{w-2}$ and $f^*_{w-1}p_{w-1}$ turns the solution into one for which 
$f^*_{w-2}p_{w-2} \geq f^*_{w-1}p_{w-1} > f^*_wp_w$ which is, again, contradictory to the fact that the solution maximizes the left-hand-size. Continuing this argument gives \cref{obs:solution-format}.

We next present the two following lemmas that we will prove later in \cref{sec:proof-of-aux-lemmas}.

\begin{lemma}\label{lem:alg1}
Let $f_1, \cdots, f_q$ and $p_1, \cdots, p_q$ be positive numbers for which 
$\sum_{i=1}^q f_i = F, \quad \sum_{i=1}^q p_i = P$
and 
$f_1 p_1 \geq f_2 p_2\geq \cdots \geq f_q p_q$. Then 
$$f_q p_q \leq \frac{FP}{q^2}$$
and equality is attained only at $f_i = \frac{F}{q}$ and $p_i = \frac{P}{q}$ for all $i\in\{1, 2, \cdots, q\}$. 
\end{lemma}

\begin{lemma}\label{lem:alg2}
Let $f_1, \cdots, f_q$ and $p_1, \cdots, p_q$ be positive variables with constraints 
$\sum_{i=1}^q f_i = F$,  $\sum_{i=1}^q p_i = P$, 
$f_1 p_1 \geq f_2 p_2\geq \cdots \geq f_q p_q$, and $f_1p_1 \leq m$ for some constant $m$. Then, the largest possible value for 
$\sum_{i=1}^q f_i p_i$ is:
$$f_{\max}(F, P, m)=\left\{
\begin{array}{l l}
FP & \textnormal{if $FP \leq m$} \\
um+(\sqrt{FP}-u\sqrt{m})^2
& \textnormal{if $\frac{FP}{(u+1)^2} \leq m < \frac{FP}{u^2}$ for $u=1, 2, \cdots, q-1$}\\
mq & \textnormal{if $m < \frac{FP}{q^2}$}
\end{array}
\right. $$
\end{lemma}

We claim that if one fixes the two quantities $f^*_wp_w = \alpha$ and $\sum_{j=1}^w p_j = \beta$, then using observation 1 and \cref{lem:alg1,lem:alg2}, the maximum value of the two terms in the statement of the theorem can be written in terms of $\alpha$ and $\beta$. Note that with $f^*_wp_w=\alpha$, both expressions are maximized when $\sum_{j=w+1}^q f^*_j p_j$ is maximized and according to \cref{lem:alg2}, that happens when 
$(\sum_{i=w+1}^q f^*_i)(\sum_{i=w+1}^q p_i) = (\sum_{i=w+1}^q f^*_i)(1-\beta)$ is maximized or equivalently $\sum_{i=1}^w f^*_i$ is as small as possible.

Now, note that for $j \leq w$ all $f^*_j p_j$'s are larger than or equal to $\alpha$. Then according to \cref{lem:alg1}, 
$\frac{(\sum_{i=1}^w f^*_i)\times \beta}{w^2} \geq \alpha \Rightarrow \sum_{i=1}^w f^*_i \geq \frac{\alpha w^2}{\beta}$.

All in all, the above-mentioned observations and lemmas boil down the two parts of theorem statement to the following:

For any $\alpha, \beta \in [0, 1]$ where $\frac{\alpha w^2}{\beta} \leq 1$:
\begin{enumerate}
    \item If $z$ is even,
    $$\frac{3z+2}{2}\alpha + (2z+1)f_{\max}\left(1-\frac{\alpha w^2}{\beta}, 1-\beta, \alpha\right) \leq 1 $$

    \item If $z$ is odd,
    $$\frac{z+1}{2}\alpha + (2z+1)f_{\max}\left(1-\frac{\alpha w^2}{\beta}, 1-\beta, \alpha\right) \leq 1 $$
\end{enumerate}
    
Note that to maximize $f_{\max}$ term for a given $\alpha$, one needs to maximize $\left(1-\frac{\alpha w^2}{\beta}\right)\left(1-\beta\right)$. This is attained with the following choice of $\beta = \sqrt{\alpha}w$. With this choice of $\beta$ we have

\begin{eqnarray*}
f_{\max}\left(1-\sqrt{\alpha}w, 1-\sqrt{\alpha}w, \alpha\right)
&=&\left\{
\begin{array}{l l}
(1-\sqrt{\alpha}w)^2 & \textnormal{if $(1-\sqrt{\alpha}w)^2 \leq \alpha$} \\ \\
u\alpha+(1-w\sqrt{\alpha}-u\sqrt{\alpha})^2
& \textnormal{if $\frac{(1-\sqrt{\alpha}w)^2}{(u+1)^2} \leq \alpha < \frac{(1-\sqrt{\alpha}w)^2}{u^2}$}\\
& \textnormal{  for $1\leq u\leq q-w$}\\
\alpha q & \textnormal{if $\alpha < \frac{(1-\sqrt{\alpha}w)^2}{(q-w)^2}$}
\end{array}
\right.\\
&=&\left\{
\begin{array}{l l}
(1-\sqrt{\alpha}w)^2 & \textnormal{if $\alpha\in\left[\frac{1}{(w+1)^2}, \frac{1}{w^2}\right]$} \\ \\
u\alpha+(1-(w+u)\sqrt{\alpha})^2 & \textnormal{if $\alpha\in\left[\frac{1}{(w+u+1)^2}, \frac{1}{(w+u)^2}\right)$}\\
 & \textnormal{  for $1\leq u\leq q-w$}\\
\alpha q & \textnormal{if $\alpha < \frac{1}{q^2}$}
\end{array}
\right.
\end{eqnarray*}

Note that we require that $\beta \leq 1 \Rightarrow  \alpha \leq \frac{1}{w^2}$. Therefore in the second line the regions for $\alpha$ are truncated at $\frac{1}{w^2}$.

As the next step, we plug in the above description for $f_{\max}$ into each of the two terms and derive a piece-wise characterization of them based on $\alpha$. 

\begin{enumerate}
    \item If $z$ is even,

    \begin{eqnarray*}
    LHS=\left\{
\begin{array}{l l}
\frac{3z+2}{2}\alpha + (2z+1)(1-\sqrt{\alpha}w)^2 & \textnormal{if $\alpha\in\left[\frac{1}{(w+1)^2}, \frac{1}{w^2}\right]$} \\ \\
\frac{3z+2}{2}\alpha + (2z+1)\left[u\alpha+(1-(u+w)\sqrt{\alpha})^2\right] & \textnormal{if $\alpha\in\left[\frac{1}{(w+u+1)^2}, \frac{1}{(w+u)^2}\right)$}\\
&\textnormal{ for $u=1, 2, \cdots, q-w$}\\
\frac{3z+2}{2}\alpha + (2z+1) \alpha q & \textnormal{if $\alpha < \frac{1}{q^2}$}
\end{array}
\right.
    \end{eqnarray*}
        
    Note that this function is continuous. The derivative in $\alpha < \frac{1}{q^2}$ region is positive meaning that the function is increasing in that region.

For the region $\alpha\in\left[\frac{1}{(w+1)^2}, \frac{1}{w^2}\right]$, 



$$\frac{\partial^2}{\partial \alpha^2}\left[\frac{3z+2}{2}\alpha + (2z+1)(1-\sqrt{\alpha}w)^2\right] = (z/2+1)(z+1/2)\alpha^{-3/2} > 0$$

Therefore, the function is concave in this region; giving that the maximum value in this region is obtained either at $\frac{1}{(w+1)^2}$ or $\frac{1}{w^2}$. Note that we can easily exclude $\frac{1}{w^2}$ as LHS function has a value of zero there.


We now analyze the derivative for the regions of form
$\alpha\in\left[\frac{1}{(w+u+1)^2}, \frac{1}{(w+u)^2}\right]$

\begin{eqnarray*}
 \frac{\partial}{\partial \alpha}\left[\frac{3z+2}{2}\alpha + (2z+1)\left(u\alpha+(1-(u+w)\sqrt{\alpha})^2\right)\right]\\ =
\frac{3z+2+(4z+2)((u+w)^2+u)}{2}-\frac{(u+w)(2z+1)}{\sqrt{\alpha}}
\end{eqnarray*}
and hence, 
$$ \frac{\partial^2}{\partial \alpha^2}\left[\frac{3z+2}{2}\alpha + (2z+1)\left(u\alpha+(1-(u+w)\sqrt{\alpha})^2\right)\right] =
(u+w)(2z+1) \alpha^{-3/2}
$$
and is always positive. Giving that within each region of form 
$\alpha\in\left[\frac{1}{(w+u+1)^2}, \frac{1}{(w+u)^2}\right]$
the expression is concave and attains no local maximum. The above observations along with the fact that this piece-wise function is continuous, gives that the global maximum is necessarily of the form
$\alpha = \frac{1}{(w+u+1)^2}$ for some $u=0, 1, 2, \cdots, q-w$. Note that at such point the value of LHS is 
\begin{eqnarray*}
LHS(u) &=& \frac{3z+2}{2}\alpha + (2z+1)\left[
u\alpha + (1-(w+u)\sqrt{\alpha})^2\right] \Bigg|_{\alpha = \frac{1}{(w+u+1)^2}}\\
&=& \frac{3z+2}{2(w+u+1)^2} + (2z+1)\frac{u+1}{(w+u+1)^2}\\
&=& \frac{3z+2+2(2z+1)(u+1)}{2(w+u+1)^2}
= \frac{7z+4+2(2z+1)u}{2(w+u+1)^2}
\end{eqnarray*}

To find the optimum $u$, we take derivative with respect to $u$. 
\begin{eqnarray*}
&&\frac{\partial}{\partial u} LHS(u) = 0\\
&\Leftrightarrow& 2(2z+1)2(w+u+1)^2 - (7z+4+2(2z+1)u)4(w+u+1) = 0\\
&\Leftrightarrow& (2z+1)(w+u+1) - (7z+4+2(2z+1)u) = 0\\
&\Leftrightarrow& -u(2z+1) + (2z+1)(z/2+2) - (7z+4) = 0\\
&\Leftrightarrow& u = \frac{(2z+1)(z/2+2) - (7z+4)}{2z+1}
= \frac{z^2 + \frac{9}{2} z + 2 - 7z - 4}{2z+1}
\\
&\Leftrightarrow& u = \frac{(2z+1)(z/2+2) - (7z+4)}{2z+1}
= \frac{z^2 - \frac{5}{2} z -2}{2z+1}=\frac{z}{2}-\frac{3z+2}{2z+1}
\end{eqnarray*}
Note that the term $\frac{3z+2}{2z+1}$ is always between 1 and 2. Hence, the maximum is achieved either at $u=\frac{z}{2}-1$ or $u=\frac{z}{2}-2$. We simply compute $LHS(u)$ for both of these values to obtain the maximum.

$$LHS\left(\frac{z}{2}-1\right) = \frac{7z+4+2(2z+1)(z/2-1)}{2(z+1)^2}
= \frac{2z^2+4z+2}{2(z+1)^2} = 1$$
and
$$LHS\left(\frac{z}{2}-2\right) = \frac{7z+4+2(2z+1)(z/2-2)}{2z^2} = 
\frac{2z^2}{2z^2} = 1$$
meaning that, indeed, the maximum achievable value for even $z$ is 1. This finishes the proof for even $z$s. The maximum value 1 can be achieved by $f^*_1=\cdots = f^*_m = \frac{1}{m} = p_1=\cdots = p_m$ and all other values equal to zero for $m = z$ or $z+1$.

    \item If $z$ is odd,
    \begin{eqnarray*}
    LHS=\left\{
\begin{array}{l l}
\frac{z+1}{2}\alpha + (2z+1)(1-\sqrt{\alpha}w)^2 & \textnormal{if $\alpha\in\left[\frac{1}{(w+1)^2}, \frac{1}{w^2}\right]$} \\ \\
\frac{z+1}{2}\alpha + (2z+1)\left[u\alpha+(1-(u+w)\sqrt{\alpha})^2\right] & \textnormal{if $\alpha\in\left[\frac{1}{(w+u+1)^2}, \frac{1}{(w+u)^2}\right)$}\\
& \textnormal{  for $u=1, 2, \cdots, q-w$}\\
\frac{z+1}{2}\alpha + (2z+1) \alpha q & \textnormal{if $\alpha < \frac{1}{q^2}$}
\end{array}
\right.
    \end{eqnarray*}

    Note that this function is continuous. The derivative in $\alpha < \frac{1}{q^2}$ region is positive meaning that the function is increasing in that region.

    Similar to the even $z$ case, for regions $\alpha\in\left[\frac{1}{(w+1)^2}, \frac{1}{w^2}\right]$ and $\alpha\in\left[\frac{1}{(w+u+1)^2}, \frac{1}{(w+u)^2}\right]$, the second derivative is positive.

$$\frac{\partial^2}{\partial \alpha^2}\left[\frac{z+1}{2}\alpha + (2z+1)(1-\sqrt{\alpha}w)^2\right] = \frac{(z+1)(z+1/2)}{2}\alpha^{-3/2} > 0$$

$$ \frac{\partial^2}{\partial \alpha^2}\left[\frac{z+1}{2}\alpha + (2z+1)\left(u\alpha+(1-(u+w)\sqrt{\alpha})^2\right)\right] =
(u+w)(2z+1) \alpha^{-3/2}
$$

Meaning that, once again, the global maximum is attained at a point necessarily of the form
$\alpha = \frac{1}{(w+u+1)^2}$ for some $u=0, 1, 2, \cdots, q-w$. Note that at such point the value of LHS is 

\begin{eqnarray*}
LHS(u) &=& \frac{z+1}{2}\alpha + (2z+1)\left[
u\alpha + (1-(w+u)\sqrt{\alpha})^2\right] \Bigg|_{\alpha = \frac{1}{(w+u+1)^2}}\\
&=& \frac{z+1+2(2z+1)(u+1)}{2(w+u+1)^2}
= \frac{5z+3+2(2z+1)u}{2(w+u+1)^2}
\end{eqnarray*}
\end{enumerate}

To find the optimum $u$, we take derivative with respect to $u$. 
\begin{eqnarray*}
&&\frac{\partial}{\partial u} LHS(u) = 0\\
&\Leftrightarrow& 2(2z+1)2(w+u+1)^2 - (5z+3+2(2z+1)u)4(w+u+1) = 0\\
&\Leftrightarrow& (2z+1)(w+u+1) - (5z+3+2(2z+1)u) = 0\\
&\Leftrightarrow& -u(2z+1) + (2z+1)\frac{z+3}{2} - (5z+3) = 0\\
&\Leftrightarrow& u = \frac{(2z+1)(z+3)/2 - (5z+3)}{2z+1}
= \frac{z^2 + \frac{7}{2} z + 3/2 - 5z - 3}{2z+1}
\\
&\Leftrightarrow& u = \frac{z^2 - \frac{3}{2} z -\frac{3}{2}}{2z+1}=\frac{z-1}{2}-\frac{z+1}{2z+1}
\end{eqnarray*}

Note that the term $\frac{z+1}{2z+1}$ is always between 0 and 1. Hence, the maximum is achieved either at $u=\frac{z-1}{2}$ or $u=\frac{z-3}{2}$. We simply compute $LHS(u)$ for both of these values to obtain the maximum.

$$LHS\left(\frac{z-1}{2}\right) = \frac{5z+3+2(2z+1)(z-1)/2}{2(z+1)^2}
= \frac{2z^2+4z+2}{2(z+1)^2} = \frac{z^2+2z+1}{z^2+2z+1} = 1$$

and

$$LHS\left(\frac{z-3}{2}\right) = \frac{5z+3+2(2z+1)(z-3)/2}{2z^2}
= \frac{2z^2}{2z^2} = 1$$

meaning that, indeed, the maximum achievable value for odd $z$ is 1. This finishes the proof. The maximum value 1 in the case of odd $z$ can be achieved by setting $f^*_1=\cdots = f^*_m = \frac{1}{m} = p_1=\cdots = p_m$ and all other values equal to zero for $m = z$ or $z+1$.
\popQED{}

\subsubsection{Proof of Auxiliary \cref{lem:alg1,lem:alg2}}\label{sec:proof-of-aux-lemmas}
\begin{proof}[{\bf Proof of \cref{lem:alg1}}]
We prove this by induction on $q$. For the base case of $q=1$ correctness is trivial. For any $q>1$, we want to find the $f_q$ and $p_q$ that maximize $f_qp_q$ and for which an appropriate $f_1, \cdots, f_{q-1}$ and $p_1, \cdots, p_{q-1}$ exists. Note that 
$\sum_{i=1}^{q-1} f_i = 1-f_q$
and
$\sum_{i=1}^{q-1} p_i = 1-p_q$. 
Therefore, by the induction hypothesis, the largest possible amount that $f_{q-1}p_{q-1}$ can take would be
$\frac{(F-f_q)(P-p_q)}{(q-1)^2}$. This gives that a pair $(f_q, p_q)$ are feasible in equations described in the lemma's statement if and only if
$f_qp_q \leq \frac{(F-f_q)(P-p_q)}{(q-1)^2}$.

Note that
\begin{eqnarray*}
&&f_qp_q \leq \frac{(F-f_q)(P-p_q)}{(q-1)^2}\\
&\Rightarrow& f_q\left(p_q + \frac{P-p_q}{(q-1)^2}\right) \leq
\frac{F(P-p_q)}{(q-1)^2}\\
&\Rightarrow& f_q \leq
\frac{F(P-p_q)}{p_q(q-1)^2 + P-p_q}\\
&\Rightarrow& f_qp_q \leq
\frac{F(P-p_q)p_q}{p_q(q^2-2q) + P}
\end{eqnarray*}

We know determine the maximum value of the right hand side over the choice of $p_q$ by setting the derivative to zero.
\begin{eqnarray*}
&&\frac
{(FP-2Fp_q)(p_q(q^2-2q) + P) - F(P-p_q)p_q(q^2-2q)}
{(p_q(q^2-2q) + P)^2} = 0\\
&\Rightarrow&
P^2
-2P p_q 
- (q^2-2q) p_q^2 = 0\\
&\Rightarrow&
p_q = \frac{-2P \pm \sqrt{4P^2+4P^2(q^2-2q)}}{2(q^2-2q)}
= \frac{-P \pm \sqrt{P^2(q^2-2q+1)}}{q^2-2q}\\
&&
= \frac{-P \pm P(q-1)}{q^2-2q}
\end{eqnarray*}

The only positive solution is 
$p_q = \frac{P}{q}$ that yields $f_qp_q = \frac{FP}{q^2}$ with $f_q = \frac{F}{q}$. Note that by the induction hypothesis, this is obtained only when
$p_i = \frac{P-p_q}{q-1}=\frac{P}{q}$ and $f_i = \frac{F-f_q}{q-1}=\frac{F}{q}$ for all $i=1, 2, \cdots, q-1$.
\end{proof}

\begin{proof}[{\bf Proof of \cref{lem:alg2}}]
We start with the simple observation that in any optimal solution in which $f_1 p_1 = \min(m, FP)$. Assume for the sake of contradiction that this is not the case. Let $j$ be the smallest integer such that $f_j p_j > 0$. Clearly, either $f_1 > f_j$ or $p_1 > p_j$. Without loss of generality assume that the former holds. Then, it is easy to verify that there exists a small enough $\eps > 0$ such that reducing $f_j$ by $\eps$ and increasing $f_1$ by $\eps$ yields a strictly larger solution and contradicts the optimality assumption.

Having this observation, we prove the lemma by induction over $q$. As the basis of the induction, take the case where $q=2$. If $m \geq PQ$, then using the above-mentioned observation, setting $f_1 = F$, $p_1 = P$, and the rest of the variables to zero yields the optimal solution. Otherwise, the observation rules that $f_1$ and $p_1$ must be chosen such that $f_1 p_1 = m$. A straight forward calculation shows that with the following choice of $f_1$ and $p_1$, $f_1 p_1 = f_2 p_2 = m$ that is trivially an optimal solution.
$$ f_1 = \frac{FP+\sqrt{F^2P^2-4mFP}}{2P}, \quad p_1 = \frac{FP-\sqrt{F^2P^2-4mFP}}{2F}$$ 
$$ f_2 = \frac{FP-\sqrt{F^2P^2-4mFP}}{2P}, \quad p_2 = \frac{FP+\sqrt{F^2P^2-4mFP}}{2F}$$ 

For the induction step, assume that the lemma holds for $q-1$. Once again we use the observation to determine $f_1$ and $p_1$ first. If $FP \leq m$, setting $f_1=F$, $p_1 = P$, and all other values to zero gives the optimal solution. Otherwise, we have to choose $f_1$ and $p_1$ such that $f_1 p_1 = m$. We can use the induction hypothesis for $q'=q-1$ to set the rest of the variables with parameters $m' = m$, $F' = F-f_1$, and $P' = P- p_1$. Note that $f_{\max}$ is actually a function of $FP$ and not $F$ and $P$. Therefore, in the optimal solution $f_1$ and $p_1$ are chosen such that $f_1p_1=m$ and $(F-f_1)(P-p_1) = FP + m -f_1P -p_1F$ is maximized, or equivalently, $f_1P + p_1F$ is minimized. Note that $f_1P + p_1F = f_1P + \frac{mF}{f_1}$. Hence one has to choose $f_1 = \sqrt{\frac{mF}{P}}$ and $p_1 = \sqrt{\frac{mP}{F}}$.

With this choice for $f_1$ and $p_1$, 
$F'P' = FP + m - 2\sqrt{mFP} = (\sqrt{FP}-\sqrt{m})^2$. Note that if 
$m < \frac{FP}{u^2} \Leftrightarrow u^2 < \frac{FP}{m} \Leftrightarrow u < \frac{\sqrt{FP}}{\sqrt{m}} = \frac{\sqrt{F'P'}}{\sqrt{m}}+1 \Leftrightarrow m < \frac{F'P'}{(u-1)^2}$.

Hence, if $\frac{FP}{(u+1)^2} \leq m < \frac{FP}{u^2}$ for some $u=2,\cdots,q-2$, then $\frac{F'P'}{u^2} \leq m < \frac{F'P'}{(u-1)^2}$ and 
$f_{\max}(F, P, m) = m + (u-1)m + (\sqrt{F'P'}-(u-1)\sqrt{m})^2 = um + (\sqrt{FP} - m)^2$.

If $\frac{FP}{4} \leq m < FP$, $f_{\max} = f_1p_1 + f_2p_2 = m + (F-f_1)(P-p_1) = m + (\sqrt{FP}-\sqrt{m})^2$.

Finally, if $m<\frac{FP}{q^2}\Leftrightarrow m<\frac{(\sqrt{F'P'}+\sqrt{m})^2}{q^2} \Leftrightarrow m < \frac{F'P'}{(q-1)^2}$ and, therefore, $f_{\max} = m + m(q-1) = mq$.
\end{proof}
}

\newpage

\bibliographystyle{plain}
\bibliography{bibliography}

\end{document}